\pdfoutput=1
\documentclass[a4paper,UKenglish]{lipics-v2019}

\title{Aperiodic Weighted Automata and Weighted First-Order Logic}

\author{Manfred Droste}{Institut für Informatik, Universität
Leipzig}{droste@informatik.uni-leipzig.de}{}
{Partly supported by a visiting professorship at ENS Paris-Saclay}

\author{Paul Gastin}{LSV, ENS Paris-Saclay, CNRS, Universit{\'e}
Paris-Saclay}{paul.gastin@ens-paris-saclay.fr}{https://orcid.org/0000-0002-1313-7722}
{Partly supported by the DFG Research Training Group QuantLA.}

\authorrunning{M.\ Droste, P.\ Gastin} 

\Copyright{Manfred Droste, Paul Gastin}



\relatedversion{An extended abstract of the paper appeared at MFCS'19. \\ \url{https://doi.org/10.4230/LIPIcs.MFCS.2019.76}}

\nolinenumbers 
\hideLIPIcs  

\usepackage[pdflatex,recompilepics=false]{gastex}
\gasset{frame=false,Nw=5,Nh=5,loopdiam=4}

\usepackage{amsmath}
\usepackage{amsfonts}
\usepackage{amssymb}
\usepackage{xspace}

\usepackage{graphicx}

\theoremstyle{plain}
\theoremstyle{remark}
\newtheorem{myremark}{Remark}

\newcommand{\A}{\ensuremath{\mathcal{A}}\xspace}
\newcommand{\B}{\ensuremath{\mathcal{B}}\xspace}
\newcommand{\AphiVp}{\ensuremath{\mathcal{A}_{\varphi,\Variables'}}\xspace}
\newcommand{\BphiV}{\ensuremath{\mathcal{B}_{\varphi,\Variables}}\xspace}
\newcommand{\ApsiV}{\ensuremath{\mathcal{A}_{\Psi,\Variables}}\xspace}

\newcommand\N{\ensuremath{\mathbb{N}}\xspace}
\newcommand\R{\ensuremath{\mathbb{R}}\xspace}

\newcommand{\FO}{\ensuremath{\mathsf{FO}}\xspace}
\newcommand{\stepFO}{\ensuremath{\mathsf{step\text{-}wFO}}\xspace}
\newcommand{\wFO}{\ensuremath{\mathsf{wFO}}\xspace}
\newcommand{\wA}{\ensuremath{\mathsf{wA}}\xspace}

\newcommand{\coreMSO}{\ensuremath{\mathsf{core\text{-}wMSO}}\xspace}
\newcommand{\wMSO}{\ensuremath{\mathsf{wMSO}}\xspace}
\newcommand{\True}{\top}
\newcommand{\False}{\bot}
\def\phi{\varphi}
\renewcommand{\bar}[1]{\overline{#1}}

\newcommand\zero{\mathbf{0}}

\newcommand\Ifthenelse[3]{{#1}\,?\,{#2}:{#3}}
\newcommand\Sum[1]{\textstyle{\sum_{#1}}}
\newcommand\Prod[1]{\textstyle{\prod_{#1}}}
\newcommand\alphabet{\ensuremath{\Sigma}}
\newcommand\Variables{\ensuremath{\mathcal V}}
\newcommand\Weights{\mathsf{R}}
\newcommand\Multiset[1]{\N\langle #1 \rangle}
\newcommand\emptymultiset{\emptyset}
\newcommand\multiset[1]{\{\!\{#1\}\!\}}
\newcommand\stepmultiset[1]{{#1}}
\newcommand\stepsem[1]{[\![ #1 ]\!]}
\newcommand\sem[1]{[\![ #1 ]\!]}
\newcommand\usem[1]{\{\!| #1 |\!\}}

\newcommand\Aggr{\mathsf{aggr}}
\newcommand\Lang[1]{\mathcal{L}(#1)}
\newcommand\pos[1]{\mathsf{pos}(#1)}
\newcommand\dom[1]{\mathsf{supp}(#1)}
\newcommand\wgt{\mathsf{wt}}
\newcommand\Tr{\mathsf{Tr}}

\begin{document}

\maketitle
\begin{abstract}
  By fundamental results of Sch\"utzenberger, McNaughton and Papert from the
  1970s, the classes of first-order definable and aperiodic languages coincide.
  Here, we extend this equivalence to a quantitative setting.  For this, weighted
  automata form a general and widely studied model.  We define a suitable notion
  of a weighted first-order logic.  Then we show that this weighted first-order
  logic and aperiodic polynomially ambiguous weighted automata have the same
  expressive power.  Moreover, we obtain such equivalence results for suitable
  weighted sublogics and finitely ambiguous or unambiguous
  aperiodic weighted automata.  Our results hold for general weight structures,
  including all semirings, average computations of costs, bounded lattices, 
  and others.
\end{abstract}

\section{Introduction}\label{sec:intro}

Fundamental results of Sch\"utzenberger, McNaughton and Papert established that
aperiodic, star-free and first-order definable languages, respectively, coincide
\cite{Schutzenberger_1965,mp71}.  In this paper, we develop such an equivalence
in a quantitative setting, i.e., for suitable notions of aperiodic weighted
automata and weighted first-order logic.

Already Sch\"utzenberger \cite{Schutzenberger61-ic} investigated weighted
automata and characterized their behaviors as rational formal power series.
Weighted automata can be viewed as classical finite automata in which the
transitions are equipped with weights.  These weights could model, e.g., the
cost, reward or probability of executing a transition.  The wide flexibility of
this automaton model soon led to a wealth of extensions and applications, cf.\
\cite{Salomaa_1978,Kuich_1986,BerstelR1988,Sakarovitch09,DrosteKV2009handbook}
for monographs and surveys.  Whereas traditionally weights are taken from a
semiring, recently, motivated by practical examples, also average and
discounted computations of weights were considered, cf.\
\cite{ChatterjeeDH10tocl,ChatterjeeDH2010lmcs}.

In the boolean setting, the seminal B\"uchi-Elgot-Trakhtenbrot theorem
\cite{Buchi60,Elgot61,Trakhtenbrot61} established the expressive equivalence of
finite automata and monadic second-order logic (MSO).  A weighted monadic
second-order logic with the same expressive power as weighted automata was
developed in \cite{DrosteG-ICALP2005,DrosteG07}.  This led to various extensions
to weighted automata and weighted logics on trees \cite{DrosteV06}, infinite
words \cite{DrosteR10}, timed words \cite{Quaas11}, pictures \cite{Fichtner11},
graphs \cite{DrosteD15}, nested words \cite{DrosteD17}, and data words
\cite{BabariDP18}, but also for more complicated weight structures including
average and discounted calculations \cite{DrosteM12} or multi-weights
\cite{DrosteP16}.  Recently, in \cite{GastinM18}, weighted MSO logic was
revisited with a more structured syntax, called \coreMSO, and shown to be
expressively equivalent to weighted automata, while permitting a uniform
approach to semirings and these more complicated weight structures.

Here, we consider the first-order fragment \wFO of this weighted logic.  It
extends the full classical boolean first-order logic quantitatively by adding
weight constants and if-then-else applications, followed by a first-order
(universal) product and then further if-then-else applications, finite sums, or
first-order (existential) sums.  We will show that its expressive power leads to
aperiodic weighted automata which, moreover, are polynomially ambiguous.
Natural subsets of connectives will correspond to unambiguous or finitely
ambiguous aperiodic weighted automata.  These various levels of ambiguity are
well-known from classical automata theory
\cite{ibarra1986,Weber_1991,KlimannLMP04,Kirsten_2008}.

Following the approach of \cite{GastinM18}, we take an \emph{arbitrary} set $R$
of weights.
A path in a weighted automaton over $R$ then has the sequence of weights of its
transitions as its value.  The \emph{abstract semantics} of the weighted
automaton is defined as the function mapping each non-empty word to the multiset
of weight sequences of the successful paths executing the given word.
Correspondingly, we will define the abstract semantics of \wFO sentences also
as functions mapping non-empty words to multisets of sequences of weights.  Our
main result will be the following.

\begin{theorem}\label{thm:main}
  Let $\Sigma$ be an alphabet and $\Weights$ a set of weights. Then the
  following classes of weighted automata and weighted first-order logics
  are expressively equivalent:
  \begin{enumerate}
    \item Aperiodic polynomially ambiguous weighted automata (\wA) and \wFO sentences,
    
    \item Aperiodic finitely ambiguous \wA and \wFO sentences 
    without first-order sums,
    
    \item Aperiodic unambiguous \wA and \wFO sentences without 
    binary or first-order sums.
  \end{enumerate}
\end{theorem}

Note that these characterizations only need aperiodicity of the underlying 
input automaton and hold without any restrictions on the weights.
The above result applies not only to the abstract semantics.  As immediate
consequence, we obtain corresponding expressive equivalence results for
classical weighted automata over arbitrary (even non-commutative) semirings, or
with average or discounted calculations of weights, or bounded lattices as in
multi-valued logics.  
All our constructions are effective. In fact, given a \wFO sentence
and deterministic aperiodic automata for its boolean subformulas,
we can construct an equivalent aperiodic weighted automaton 
of exponential size.
We give typical examples for our constructions.  We also show that the class
of arbitrary aperiodic weighted automata and its subclasses of polynomially
resp.\ finitely ambiguous or unambiguous weighted automata form a proper
hierarchy for each of the following semirings: natural numbers
$\N_{+,\times}$, the max-plus-semiring $\N_{\max,+}$ and the min-plus semiring
$\N_{\min,+}$.
Results are summarized in the Table~\ref{tbl:summary}. The second column lists existing 
characterizations for general weighted automata. The last column contains our 
results concerning \emph{aperiodic} weighted automata. Examples referred to in 
the table separate the classes.

\begin{table}[t]
  \centering
  \begin{tabular}{|c|l|l|l|}
    \hline
    ambiguity & \multicolumn{1}{c|}{general WA} & \multicolumn{1}{c|}{aperiodic WA} \\
    \hline
    exponential & \wMSO: \cite{DrosteG-ICALP2005,DrosteG07} 
    & {Ex.~\ref{ex:plustimes-exponentially-ambiguous} for $\N_{+,\times}$,
    Ex.~\ref{ex:maxplus-exponentially-ambiguous} for $\N_{\max,+}$}
    \\
    & {Ex.~\ref{ex:minplus-exponentially-ambiguous} for $\N_{\min,+}$: \cite{mazowieckiR2018}}
    & {Ex.~\ref{ex:minplus-exponentially-ambiguous} for $\N_{\min,+}$}
    \\ \hline
    & $\wMSO$ without $\sum_{X}$: \cite{KreutzerR2013}  
    &  \wFO: Thm.~\ref{thm:main2}, Thm.~\ref{thm:tmp12}  
    \\
    polynomial & 
    & {Ex.~\ref{ex:plustimes-polynomially-ambiguous} for $\N_{+,\times}$,
    Ex.~\ref{ex:maxplus-polynomially-ambiguous} for $\N_{\max,+}$}
    \\
    & {Ex.~\ref{ex:minplus-polynomially-ambiguous} for $\N_{\min,+}$: \cite{Kirsten_2008,mazowieckiR2018}}
    & {Ex.~\ref{ex:minplus-polynomially-ambiguous} for $\N_{\min,+}$}
    \\ \hline
    & $\wMSO$ without $\sum_{X},\sum_{x}$: \cite{KreutzerR2013}  
    &  $\wFO$ without $\sum_{x}$: 
    Cor.~\ref{cor:aperiodic-unambiguous-to-logic}(2), Thm.~\ref{thm:tmp12}  \\
    finite & & Ex.~\ref{ex:plustimes-finitely-ambiguous} for $\N_{+,\times}$ \\
    & Ex.\ref{ex:minmaxplus-finitely-ambiguous}: $\N_{\max,+}$ 
    \cite{KlimannLMP04}, $\N_{\min,+}$ \cite{mazowieckiR2018}  
    & Ex.~\ref{ex:minmaxplus-finitely-ambiguous} for $\N_{\max,+}$
    and $\N_{\min,+}$
    \\ \hline
    unambiguous & $\wMSO$ without $\sum_{X},\sum_{x},{+}$: \cite{KreutzerR2013}  &  $\wFO$ without 
    $\sum_{x},{+}$: 
    Cor.~\ref{cor:aperiodic-unambiguous-to-logic}(1), Thm.~\ref{thm:tmp12}  \\
    \hline
  \end{tabular}
  \caption{Summary of our main results.}
  \label{tbl:summary}
\end{table}

It should be noticed that standard constructions used to establish
equivalence between automata and MSO logic cannot be applied.  Indeed, starting
from an automaton $\A$, one usually constructs an \emph{existential} MSO sentence
where the existential set quantifications are used to guess an accepting run and
the easy first-order kernel is used to check that this guess indeed defines an
accepting run.  Here, we cannot use quantifications $\exists X$ over set
variables $X$, or their weighted equivalent $\sum_{X}$.  Instead, we take
advantage of the fine structure of possible paths of polynomially ambiguous
automata, namely the fact that it must be unambiguous on strongly connected
components (SCC-unambiguous), as employed for different goals already in
\cite{ibarra1986,Weber_1991}.  We first give a new construction of a
$\wFO$ sentence without sums starting from an aperiodic and 
unambiguous automaton. Then, we extend the construction to 
polynomially ambiguous aperiodic automata using first-order sums $\Sum{x}$ to 
guess positions where the run switches between the unambiguous SCCs.
For part 2 of Theorem~\ref{thm:main}, we also prove that for each \emph{aperiodic}
finitely ambiguous weighted automaton we can construct finitely many \emph{aperiodic}
unambiguous weighted automata whose disjoint union has the same semantics.

Again, for the implication from weighted formulas to weighted automata, we
cannot simply use standard constructions which crucially rely on the fact that
functions defined by weighted automata are closed under morphic images.  This
was used to handle first-oder sums $\Sum{x}$ and second-order sums $\Sum{X}$,
but also in the more involved proof for the first-order product $\Prod{x}$
applied to finitely valued weighted automata.  But it is well-known that
aperiodic languages are not closed under morphic images.  Handling the
first-order product $\Prod{x}$ requires a completely new and highly non-trivial
proof preserving aperiodicity properties.

\smallskip\noindent
\textbf{Related work}.
In \cite{KreutzerR2013}, polynomially ambiguous, finitely ambiguous
and unambiguous weighted automata (without assuming aperiodicity) over commutative semirings
were shown to be expressively equivalent to suitable fragments of weighted
monadic second order logic. This was further extended  in \cite{Paul_2016}
to cover polynomial degrees and weighted tree automata.

A hierarchy of these classes of weighted automata (again without assuming aperiodicity)
over the max-plus semiring was described in \cite{KlimannLMP04}.
As a consequence of pumping lemmas for weighted automata, a similar hierarchy
was obtained in \cite{mazowieckiR2018} for the min-plus semiring.

We note that in \cite{DrosteG07,DrosteV12}, an equivalence result for
full weighted first-order logic was given, but only for very particular classes
of semirings or strong bimonoids as weight structures.

A characterization of the full weighted first-order logic with transitive closure
by weighted pebble automata was obtained in \cite{BolligGMZ14}.
An equivalence result for fragments of weighted first-order logic, weighted LTL
and weighted counter-free automata over the max-plus semiring with discounting
was given in \cite{Mandrali_2014}.
Various further equivalences to boolean first-order definability of languages were
described in the survey \cite{DiGa08Thomas}. Due to its possible applications
for quantitative verification questions, it remains a challenging problem to develop
a weighted linear temporal logic for general classes of semirings with sufficiently
large expressive power.

\section{Preliminaries}\label{sec:prelim}

A non-deterministic automaton is a tuple $\A=(Q,\Sigma,\Delta)$ where $Q$ is a
finite set of states, $\Sigma$ is a finite alphabet, and $\Delta\subseteq
Q\times\Sigma\times Q$ is the set of transitions.  The automaton $\A$ is
complete if $\Delta(q,a)\neq\emptyset$ for all $q\in Q$ and $a\in\Sigma$.
A run $\rho$ of $\A$ is a nonempty sequence of transitions
$\delta_1=(p_1,a_1,q_1)$, $\delta_2=(p_2,a_2,q_2)$, \ldots,
$\delta_n=(p_n,a_n,q_n)$ such that $q_i=p_{i+1}$ for all $1\leq i<n$.  We say
that $\rho$ is a run from state $p_1$ to state $q_n$ and that $\rho$ reads, or 
has label, the word $a_1a_2\cdots a_n\in\Sigma^+$.  We denote by
$\Lang{\A_{p,q}}\subseteq\Sigma^*$ the set of labels of runs of $\A$ from $p$ to 
$q$. When $p=q$, we include the empty word $\varepsilon$ in $\Lang{\A_{p,q}}$ 
and say that $\varepsilon$ labels the empty run from $p$ to $p$.

An automaton with accepting conditions is a tuple $\A=(Q,\Sigma,\Delta,I,F)$
where $(Q,\Sigma,\Delta)$ is a non-deterministic automaton, $I,F\subseteq Q$ are
the sets of initial and final states respectively.  The language defined by the
automaton is $\Lang{\A}=\Lang{\A_{I,F}}=\bigcup_{p\in I,q\in F}\Lang{\A_{p,q}}$.
Subsequently, we also consider automata with several accepting sets $F,G,\ldots$ so that
the same automaton may define several languages $\Lang{\A_{I,F}}$,
$\Lang{\A_{I,G}}$, \ldots An automaton $\A=(Q,\Sigma,\Delta,I,F)$ is
\emph{deterministic} if $I=\{\iota\}$ is a singleton and the set $\Delta$ of
transitions is a partial function: for all $(p,a)\in Q\times\Sigma$ there is at
most one state $q\in Q$ such that $(p,a,q)\in\Delta$.

Next, we consider degrees of ambiguity of automata.  A run in an automaton is
\emph{successful}, if it leads from an initial to a final state.  The automaton
$\A$ is called \emph{polynomially ambiguous} if there is a polynomial $p$ such
that for each $w\in\Sigma^+$ the number of successful paths in $\A$ for $w$ is
at most $p(|w|)$.  Then, $\A$ is \emph{finitely ambiguous} if $p$ can be taken to be a constant.
Further, for an integer $k\geq1$, $\A$ is $k$-ambiguous if $p=k$, and
\emph{unambiguous} means 1-ambiguous.  Notice that $k$-ambiguous implies
$(k+1)$-ambiguous.  An automaton $\A$ is at most \emph{exponentially
ambiguous}.

A non-deterministic automaton $\A=(Q,\Sigma,\Delta)$ is \emph{aperiodic} if
there exists an integer $m\geq1$, called aperiodicity index, such that for all
states $p,q\in Q$ and all words $u\in\Sigma^+$, we have $u^m\in\Lang{\A_{p,q}}$
iff $u^{m+1}\in\Lang{\A_{p,q}}$.
In other words, the non-deterministic automaton $\A$ is aperiodic iff its
transition monoid $\Tr(\A)$ is aperiodic.  
It is well-known that aperiodic languages coincide with
first-order definable languages, cf.~\cite{Schutzenberger_1965,mp71,DiGa08Thomas}.

The syntax of first-order logic is given in Section~\ref{sec:wFO}
\eqref{eq:fo}.  The semantics is defined by structural
induction on the formula and requires an interpretation of the free variables.
Let $\Variables=\{y_1,\ldots,y_n\}$ be a finite set of first-order variables.
Given a nonempty word $u\in\Sigma^+$, we let $\pos{u}=\{1,\ldots,|u|\}$ be the
set of positions of $u$.  A valuation or interpretation is a map
$\sigma\colon\Variables\to\pos{u}$ assigning positions of $u$ to variables in
$\Variables$.  For a first-order formula $\varphi$ having free variables
contained in $\Variables$, we write $u,\sigma\models\varphi$ when the word $u$
satisfies $\varphi$ under the interpretation defined by $\sigma$.
When $\varphi$ is a \emph{sentence}, the valuation $\sigma$ is not needed and we
simply write $u\models\varphi$.

We extend the classical semantics by defining when the empty word $\varepsilon$
satisfies a sentence.  We have $\varepsilon\models\True$ and if $\forall x\psi$
is a sentence then $\varepsilon\models\forall x\psi$.  The semantics
$\varepsilon\models\varphi$ is extended to all sentences $\varphi$ since they are
boolean combinations of the basic cases above.  Notice that if $\varphi$ has
free variables then $\varepsilon\models\varphi$ is not defined. When $\varphi$ 
is a sentence we denote by $\Lang{\varphi}\subseteq\Sigma^{*}$ the set of words 
satisfying $\varphi$. Notice that $\Lang{\forall x\False}=\{\varepsilon\}$ 
where $\False=\neg\True$.

\begin{theorem}[\cite{Schutzenberger_1965,mp71,DiGa08Thomas}]\label{thm:ap2fo}
  Let $\A$ be an aperiodic non-deterministic automaton.  For
  all states $p,q$ of $\A$ we can construct a first-order sentence
  $\varphi_{p,q}$ such that $\Lang{\A_{p,q}}=\Lang{\varphi_{p,q}}$.
\end{theorem}

For the converse of Theorem~\ref{thm:ap2fo}, we need a stronger statement to
deal with formulas having free variables.  As usual, we encode a pair
$(u,\sigma)$ where $u\in\Sigma^+$ is a nonempty word and
$\sigma\colon\Variables\to\pos{u}$ is a valuation by a word $\bar{u}$ over the
extended alphabet $\Sigma_\Variables=\Sigma\times\{0,1\}^\Variables$.  A word
$\bar{u}$ over $\Sigma_\Variables$ is a \emph{valid} encoding if for each
variable $y\in\Variables$, its projection on the $y$-component belongs to
$0^*10^*$.  Throughout the paper, we identify a valid word $\bar{u}$ with its
encoded pair $(u,\sigma)$.

\begin{theorem}[\cite{Schutzenberger_1965,mp71,DiGa08Thomas}]\label{thm:fo2ap}
  For each \FO-formula $\varphi$ having free variables contained in
  $\Variables$, we can build a deterministic, complete and aperiodic automaton
  $\A_{\varphi,\Variables}=(Q,\Sigma_\Variables,\Delta,\iota,F,G)$ over the
  extended alphabet $\Sigma_\Variables$ such that for all words
  $\bar{u}\in\Sigma_\Variables^{+}$ we have:
  \begin{itemize}
    \item $\Delta(\iota,\bar{u}) \in F$ iff $\bar{u}$
    is a valid encoding of a pair $(u,\sigma)$ with $(u,\sigma)\models\varphi$,
  
    \item $\Delta(\iota,\bar{u}) \in G$ iff $\bar{u}$
    is a valid encoding of a pair $(u,\sigma)$ with 
    $(u,\sigma)\models\neg\varphi$,

    \item $\Delta(\iota,\bar{u}) \notin F\cup G$ otherwise, i.e., 
    iff $\bar{u}$ is not a valid encoding of a pair $(u,\sigma)$.
    \end{itemize}
\end{theorem}

Given $u\in\Sigma^{+}$ and integers $k,\ell$, we denote by $u[k,\ell]$ the
factor of $u$ between positions $k$ and $\ell$.  By convention
$u[k,\ell]=\varepsilon$ is the empty word when $\ell<k$ or $\ell=0$ or $k>|u|$.
 
We will apply the equivalence of Theorem~\ref{thm:ap2fo} to prefixes, infixes or suffixes of words.
Towards this, we use the classical \emph{relativization} of sentences.  Let
$\varphi$ be a first-order sentence and let $x,y\in\Variables$ be first-order
variables.  We define below the relativizations $\varphi^{<x}$, $\varphi^{(x,y)}$ and
$\varphi^{>y}$ so that for all words $u\in\Sigma^+$, and all positions
$i,j\in\pos{u}=\{1,\ldots,|u|\}$ we have
\begin{align*}
  u,x\mapsto i & \models \varphi^{<x} & \text{iff} && u[1,i-1] & \models \varphi \\
  u,x\mapsto i,y\mapsto j & \models \varphi^{(x,y)} & \text{iff} && u[i+1,j-1] & \models \varphi \\
  u,x\mapsto j & \models \varphi^{>x} & \text{iff} && u[j+1,|u|] & \models \varphi 
\end{align*}
Notice that, when $i=1$ or $j\leq i+1$ or $j=|u|$, the relativization is on the
empty word, this is why we had to define when $\varepsilon\models\psi$ for
sentences $\psi$.  The relativization is defined by structural induction on the
formulas as follows:
\begin{align*}
  \True^{<x} & =\True & (P_a(z))^{<x} & =P_a(z) & (y\leq z)^{<x} & = (y\leq z) \\
  (\neg\psi)^{<x} & = \neg(\psi^{<x}) & (\psi_1\wedge\psi_2)^{<x} & = 
  \psi_1^{<x} \wedge \psi_2^{<x} & (\forall z\psi)^{<x} & = \forall z(z<x 
  \implies \psi^{<x})
\end{align*}
The relativizations $\varphi^{(x,y)}$ and $\varphi^{>x}$ are defined similarly.
Notice that when $\varphi$ is a sentence, i.e., a boolean combination of 
formulas of the form $\True$ or $\forall z\psi$, then the above equivalences 
hold even when $i=1$ for $\phi^{<x}$, or when $i=|u|$ for $\phi^{>x}$, or when 
$j\leq i+1$ for $\varphi^{(x,y)}$.

\section{Weighted Automata}\label{sec:wA}

Given a set $X$, we let $\Multiset{X}$ be the collection of all finite multisets
over $X$, i.e., all functions $f\colon X \rightarrow \N$ such that $f(x) \ne 0$
only for finitely many $x \in X$.  The multiset union $f \uplus g$ of two
multisets $f,g \in \Multiset{X}$ is defined by pointwise addition of functions:
$(f \uplus g)(x) = f(x) + g(x)$ for each $x \in X$.

For a set $\Weights$ of weights, an $\Weights$-weighted automaton over $\Sigma$
is a tuple $\A=(Q,\Sigma,\Delta,\wgt)$ where $(Q,\Sigma,\Delta)$ is a
non-deterministic automaton and $\wgt\colon\Delta\to\Weights$ assigns a weight
to every transition.
The weight sequence of a run $\rho=\delta_1\delta_2\cdots\delta_n$ is  
$\wgt(\rho)=\wgt(\delta_1)\wgt(\delta_2)\cdots\wgt(\delta_n)\in\Weights^+$. The 
\emph{abstract semantics} of $\A$ from state $p$ to state $q$ is the map 
$\usem{\A_{p,q}}\colon\Sigma^+\to\Multiset{\Weights^+}$ which assigns to a 
word $u\in\Sigma^+$ the multiset of weight sequences of runs from $p$ to $q$ 
with label $u$:
\[
\usem{\A_{p,q}}(u)=\multiset{\wgt(\rho)\mid\rho \text{ is a run from $p$ to $q$ 
with label } u}\,.
\]
Notice that $\usem{\A_{p,q}}(u)=\emptymultiset$ is the empty multiset when there
are no runs of $\A$ from $p$ to $q$ with label $u$, i.e., when
$u\notin\Lang{\A_{p,q}}$.  When we consider a weighted automaton
$\A=(Q,\Sigma,\Delta,\wgt,I,F)$ with initial and final sets of states, for all
$u\in\Sigma^+$ the semantics $\usem{\A}$ is defined as the multiset union:
$\usem{\A}(u)=\biguplus_{p\in I,q\in F}\usem{\A_{p,q}}(u)$.  Hence, $\usem{\A}$
assigns to every word $u\in\Sigma^{+}$ the multiset of all weight sequences of
accepting runs of $\A$ reading $u$.
The support of $\A$ is the set of words
$u\in\Sigma^+$ such that $\usem{\A}(u)\neq\emptymultiset$, i.e.,
$\dom{\A}=\Lang{\A}$.
\begin{gpicture}[name=WA1,ignore]
  \node[Nmarks=i](1)(0,0){1}
  \node(2)(20,0){2}
  \node(3)(40,0){3}
  \node[Nmarks=f](4)(60,0){4}
  \drawloop(1){$a \mid 2$}
  \drawedge(1,2){$a \mid 1$}
  \drawloop(2){$a \mid 3$}
  \drawedge[curvedepth=2](2,3){$a \mid 5$}
  \drawloop(3){$b \mid 5$}
  \drawedge[curvedepth=2](3,2){$b \mid 3$}
  \drawedge(3,4){$b \mid 1$}
  \drawloop(4){$b \mid 2$}
\end{gpicture}
\begin{figure}[tbp]
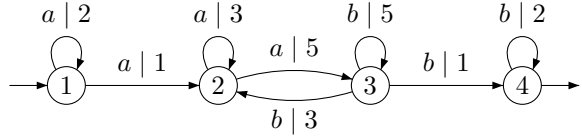

  \centerline{\gusepicture{WA1}}
  \caption{A weighted automaton, which is both aperiodic and polynomially ambiguous.}
  \protect\label{fig:WA-SCC-unambiguous}
\end{figure}

For instance, consider the weighted automaton $\A$ of 
Figure~\ref{fig:WA-SCC-unambiguous}. We have $\dom{\A}=a^+a(a+b)^*b^+$. 
Moreover, consider $w=a^m(ba)^nb^p$ with $m>1$ and $p>0$. We have 
$w\in\dom{\A}$ and
\[
\usem{\A}(w)=\multiset{2^{k-1}\cdot 1\cdot 3^{m-k-1}\cdot 5\cdot (3\cdot 5)^n\cdot 
5^{\ell-1}\cdot 1\cdot 2^{p-\ell} \mid 1\leq k<m \text{ and } 1\leq\ell\leq p} \,.
\]

A concrete semantics over semirings, or valuation monoids, or valuation 
structures can be obtained from the abstract semantics defined above by 
applying the suitable aggregation operator 
$\Aggr\colon\Multiset{\Weights^+}\to S$ as explained in \cite{GastinM18}.
For convenience, we include a short outline.

A \emph{semiring} is a structure  $(S,+,\times,0,1)$  where
$(S,+,0)$ is a commutative monoid, $(S,\times,1)$  is a monoid,
multiplication distributes over addition, and $0\times s = s \times 0 = 0$
for each  $s \in S$. If the multiplication is commutative, we say that
$S$  is \emph{commutative}. If the addition is idempotent,
the semiring is called \emph{idempotent}.
Important examples of semirings include:
\begin{itemize}
  \item the natural numbers  $\N_{+,\times}=(\N,+,\times,0,1)$  with the usual addition and multiplication;
  \item the Boolean semiring   $\B = (\{0,1\},\vee,\wedge,0,1)$;
  \item the min-plus (or tropical) semiring  $\N_{\min,+}=(\N \cup \{\infty\},\min,+,\infty,0)$;
  \item the max-plus (or arctical) semiring  $\N_{\max,+}=((\N \cup \{-\infty\},\max,+,-\infty,0)$;
  \item the semiring of languages 
  $(\mathcal{P}(\Sigma^*),\cup,\cdot,\emptyset,\{\varepsilon\})$
  where  $\cdot$  denotes concatenation of languages;
  \item the semiring of multisets of sequences
  $(\Multiset{\Weights^*},\uplus,\cdot,\emptyset,\multiset{\varepsilon})$.
  
  Here, 
  $\cdot$ denotes the concatenation of multisets (Cauchy product), cf.\
  \cite{GastinM18}.
\end{itemize}

Let  $(S,+,\times,0,1)$ be a semiring and  $\A=(Q,\Sigma,\Delta,\wgt)$ be an
$S$-weighted automaton over  $\Sigma$. The value of a run
$\rho = \delta_1 \delta_2 \cdots \delta_n$  is then defined as
$\mathsf{val}(\rho) = \wgt(\delta_1) \times \wgt(\delta_2) \times\cdots\times \wgt(\delta_n)$.
The \emph{concrete semantics} of  $\A$  is the function
$\sem{\A}\colon \Sigma^+ \rightarrow S$  given by
$\sem{\A}(w) = \sum_\rho val(\rho)$  where the sum is taken over all
successful paths  $\rho$  executing the word  $w$.

Let us define the \emph{aggregation function}
$\Aggr_{\mathsf{sp}}\colon\Multiset{\Weights^+}\to S$ by letting 
$\Aggr_{\mathsf{sp}}(f)$ be the sum over
all sequences $s_1s_2\cdots s_k$ in the multiset $f$ of the products $s_1\times
s_2\times\cdots\times s_k$ in $S$.  It follows that the concrete semantics of
$\A$ is the composition of the aggregation function and the abstract semantics
of $\A$, i.e., $\sem{\A}(w) = \Aggr_{\mathsf{sp}}(\usem{\A}(w))$ for all $w \in \Sigma^+$.
Also, the abstract semantics $\usem{\A}$ conicides with the concrete semantics
of $\A$ over the semiring of multisets of sequences over $S$, i.e., $\usem{\A} =
\sem{\A}$ (since the aggregation function is the identity function).

As another example, assume the weights of $\A$ are taken in $\R_{\geq 0} \cup
\{-\infty\})$, the weight of a run $\rho$ is computed as the average
$\mathsf{avg}(\rho) = (\wgt(\delta_1) + \cdots + \wgt(\delta_n))/n$ of the
weights in $\rho$, and the concrete semantics of $\A$ is defined for $w \in
\Sigma^+$ by $\sem{\A}(w) = \max_\rho \mathsf{avg}(\rho)$ where the maximum is
taken over all successful runs $\rho$ executing $w$, cf.\
\cite{ChatterjeeDH2010lmcs,ChatterjeeDH10tocl,DrosteM12}.  In this case, we define the
aggregation $\Aggr_{\mathsf{ma}}(M)$ of a multiset $M$ by taking the maximum of all averages
of sequences in $M$.  Again, we obtain $\sem{\A}(w) = \Aggr_{\mathsf{ma}}(\usem{\A}(w))$ for
all $w \in \Sigma^+$.  See \cite{GastinM18} for further discussion and examples.

Now, consider the natural semiring $(\N,+,\times,0,1)$ and the
sum-product aggregation operator $\Aggr_{\mathsf{sp}}$.  We continue the example
above with the automaton $\A$ of Figure~\ref{fig:WA-SCC-unambiguous} and the
word $w=a^m(ba)^nb^p$ with $m>1$ and $p>0$.  The concrete semantics is given by
\[
\sem{\A}(w)=\Aggr_{\mathsf{sp}}(\usem{\A}(w))
=\sum_{1\leq k<m} \sum_{1\leq\ell\leq p}
2^{k-1+p-\ell} 3^{m-k-1+n} 5^{n+\ell}  \,.
\]

\section{Finitely ambiguous Weighted Automata}\label{sec:fa-wa}
\begin{gpicture}[name=3ambiguous,ignore]
  \node[Nmarks=i,iangle=180](1)(0,0){1}
  \node(2)(20, 7){2}
  \node(3)(20,-7){3}
  \node(4)(40, 7){4}
  \node(5)(40,-7){5}
  \node[Nmarks=f,fangle=0](6)(60,0){6}
  \drawloop[loopangle=90](1){$a \mid 2$}
  \drawedge[curvedepth=0](1,2){$a \mid 2$}
  \drawedge[curvedepth=0,ELside=r](1,3){$a \mid 1$}
  \drawedge[curvedepth=0](2,4){$a \mid 1$}
  \drawedge[curvedepth=0,ELdist=0](2,5){$a \mid 3$}
  \drawedge[curvedepth=0,ELside=r](3,5){$a \mid 5$}
  \drawedge[curvedepth=0](4,6){$a \mid 4$}
  \drawedge[curvedepth=0,ELside=r](5,6){$b \mid 3$}
  \drawloop[loopangle=90](6){$b \mid 3$}
\end{gpicture}
In this section, we investigate finitely ambiguous weighted automata.  It
was shown in \cite{KlimannLMP04} that over the max-plus semiring $\N_{\max,+}$
they are expressively 
equivalent to finite disjoint unions of unambiguous weighted automata.
Moreover, it was proved in \cite{Sakarovitch_2009} that a $K$-valued rational
transducer can be decomposed into $K$ unambiguous transducers.  In particular
this implies that a $K$-ambiguous weighted automaton can be decomposed into $K$
unambiguous weighted automata.
Here we show that the same holds for aperiodic weighted automata.

\begin{theorem}\label{thm:finitely-ambiguous}
  Let $K\geq1$.  
  Given an aperiodic $K$-ambiguous weighted automaton $\A$, we can 
  construct aperiodic unambiguous weighted automata $\B_1,\ldots,\B_K$ such that
  $\usem{\A}=\usem{\B_1\uplus\cdots\uplus\B_K}=\usem{\B_1}\uplus\cdots\uplus\usem{\B_K}$. 
\end{theorem}

We give below a direct and simple construction which works for arbitrary
(possibly non-aperiodic) $K$-ambiguous weighted automata.  Then, we show that
our construction preserves aperiodicity.  Our proof is based on lexicographic
ordering of runs.  The proof of \cite{Sakarovitch_2009} uses lexicographic
coverings.  It would be interesting to see whether this proof also preserves
aperiodicity and to compare the complexity of the constructions.

\begin{figure}[tbp]
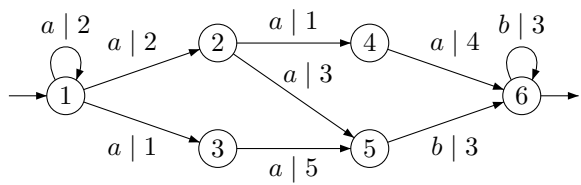

  \centerline{\gusepicture{3ambiguous}}
  \caption{A 3-ambiguous weighted automaton.}
  \protect\label{fig:3ambiguous}
\end{figure}

We first explain our construction on an example.
Consider the 3-ambiguous weighted automaton $\A$ of
Figure~\ref{fig:3ambiguous} over the alphabet $\Sigma=\{a,b\}$ and the semiring
$\N_{+,\times}$ of natural numbers.  Clearly, the support of $\A$ is
$a^{*}(a^{3}+a^{2}b)b^{*}$ and $\sem{\A}(a^{n}a^{3}bb^{p})=2^{n}\cdot(
2\cdot 1\cdot 4\cdot 3
+2\cdot 2\cdot 3\cdot 3
+2\cdot 1\cdot 5\cdot 3)\cdot 3^{p}$ for $n,p\geq0$.  
We construct in the proof an automaton $\A_{\geq3}$
which checks that $\A$ has 3 accepting runs on a given word.  Hence, we will
have $\Lang{\A_{\geq3}}=a^{*}a^{3}bb^{*}$.  To do so, $\A_{\geq3}$ will run
three copies of $\A$, make sure that the three runs are lexicographically 
ordered (to be unambiguous) and accept if the three runs are
accepting and pairwise distinct.
The set of states is $Q'=Q^{3}\times\{0,1\}^{2}$ where $Q=\{1,\ldots,6\}$ is the
set of states of $\A$.  The initial state is $(1,1,1,0,0)$ and the booleans turn to
1 when the runs differ.  The accepting state is $(6,6,6,1,1)$.  The unique
accepting run of $\A_{\geq3}$ on the word $a^{3}b^{2}$ is
$$
(1,1,1,0,0)\xrightarrow{a}(1,1,2,0,1)\xrightarrow{a}(2,3,4,1,1)
\xrightarrow{a}(5,5,6,1,1)\xrightarrow{b}(6,6,6,1,1)\xrightarrow{b}(6,6,6,1,1) \,.
$$

\begin{proof}
  First, let $\A=(Q,\Sigma,\Delta,I,F,\wgt)$ be an arbitrary weighted automaton.
  We may assume that $\A$ has a single initial state $q_0$.
  For $k\geq1$, we construct an automaton
  $\A_{\geq k}=(Q',\Sigma,\Delta',I',F')$ which accepts the set of words
  $w=a_1a_2\cdots a_n\in\Sigma^{+}$ having at least $k$ accepting runs in $\A$.
  Moreover, if $\A$ is aperiodic then so is $\A_{\geq k}$.
  
  Fix a strict total order $\prec$ on $Q$.  We write $\preceq$ for the induced
  lexicographic order on $Q^{+}$ and $\prec$ for the strict order.  A run of
  $\A$ on $w$ induces a sequence of states
  $\rho=q_0q_1q_2\cdots q_n\in Q^{+}$ with $(q_{i-1},a_i,q_i)\in\Delta$ for all
  $1\leq i\leq n$.  Overloading our terminology, such a sequence is also called 
  a run below. Runs of $\A$ on $w$ are lexicographically ordered.  For
  $0\leq j\leq n$, we denote by $\rho[j]=q_0q_1q_2\cdots q_j$ the prefix of
  length $j$ of $\rho$.
  
  The idea is that $\A_{\geq k}$ will guess $k$ runs
  $\rho^{1}\preceq\rho^{2}\preceq\cdots\preceq\rho^{k}$ of $\A$ on $w$.  For
  $1\leq \ell\leq k$, we let $\rho^{\ell}=q^{\ell}_0q^{\ell}_1q^{\ell}_2\cdots
  q^{\ell}_n$.  Now, after reading the prefix $w[j]=a_1a_2\cdots a_j$, the state
  of $\A_{\geq k}$ will consist of the tuple $(q^{1}_j,\ldots,q^{k}_j)$ of states reached
  by the prefixes $\rho^{1}[j]\preceq\cdots\preceq\rho^{k}[j]$ together with a bit
  vector $(c^{1}_j,\ldots,c^{k-1}_j)$ such that for all $1\leq\ell<k$,
  $c^{\ell}_j=1$ iff $\rho^{\ell}[j]\prec\rho^{\ell+1}[j]$.  The automaton $\A_{\geq k}$
  will accept if all states $q^{\ell}_n\in F$ are final in $\A$ and the bit vector
  contains only 1's.  This ensures that
  $\rho^{1}\prec\rho^{2}\prec\cdots\prec\rho^{k}$ are distinct accepting runs for
  $w$ in $\A$.
  
  We turn now to the formal definition of $\A_{\geq k}$.  Let
  $Q'=Q^{k}\times\{0,1\}^{k-1}$, $I'=\{q_0\}^{k}\times\{0\}^{k-1}$ and
  $F'=F^{k}\times\{1\}^{k-1}$.  We write tuples with superscripts: 
  $(\bar{q},\bar{c})\in Q'$ with $\bar{q}=(q^{1},\ldots,q^{k})$ and 
  $\bar{c}=(c^{1},\ldots,c^{k-1})$.
  Now, $((\bar{q},\bar{c}),a,(\bar{q}',\bar{c}'))$ is a
  transition of $\A_{\geq k}$ if the following conditions hold:
  \begin{itemize}
    \item $(q^{\ell},a,q'^{\ell})\in\Delta$ for all $1\leq\ell\leq k$ (the $k$ runs
    are non-deterministically guessed),
  
    \item and the bit vector is deterministically updated as follows: for
    $1\leq\ell<k$ we have either ($c^{\ell}=0$, $q'^{\ell}=q'^{\ell+1}$ and
    $c'^{\ell}=0$), or (($c^{\ell}=1$ or $q'^{\ell}\prec q'^{\ell+1}$) and
    $c'^{\ell}=1$). Notice that $c^{\ell}=0$ and $q'^{\ell+1}\prec q'^{\ell}$ 
    is not allowed.
  \end{itemize}
  When $k=1$ then the accessible part of $\A_1$ is equal to $\A$.
  We will now state formally the main properties of $\A_{\geq k}$.  
  
  \begin{claim}\label{claim:bij-run}
    For each $w\in\Sigma^{+}$, there is a bijection between the accepting runs
    $\bar{\rho}$ of $\A_{\geq k}$ on $w$ and the tuples
    $(\rho^{1},\ldots,\rho^{k})$ of accepting runs of $\A$ on $w$ such that
    $\rho^{1}\prec\cdots\prec\rho^{k}$.
  \end{claim}
  
  \begin{claimproof}
    Consider a word
    $w=a_1a_2\cdots a_n\in\Sigma^{+}$ and a run $\bar{\rho}$ of $\A_{\geq k}$ on $w$
    starting from its initial state.  Write $(\bar{q}_j,\bar{c}_j)$ the $j$th state
    of $\bar{\rho}$.  For $1\leq\ell\leq k$, let $\rho^{\ell}$ be the projection of
    $\bar{\rho}$ on the $\ell$th component:
    $\rho^{\ell}=q^{\ell}_0q^{\ell}_1q^{\ell}_2\cdots q^{\ell}_n$. Clearly, 
    $\rho^{\ell}$ is a run of $\A$ on $w$. Moreover, we can easily check by 
    induction on $0\leq j\leq n$ that for all $1<\ell\leq k$ we have 
    $\rho^{\ell}[j]=\rho^{\ell+1}[j]$ if $c^{\ell}_j=0$ and
    $\rho^{\ell}[j]\prec\rho^{\ell+1}[j]$ if $c^{\ell}_j=1$. We deduce that if 
    $\bar{\rho}$ is accepting in $\A_{\geq k}$ then each $\rho^{\ell}$ is accepting in 
    $\A$ and $\rho^{1}\prec\cdots\prec\rho^{k}$. Therefore, every word accepted by 
    $\A_{\geq k}$ admits at least $k$ accepting runs in $\A$.
    
    Conversely, assume that $w\in\Sigma^{+}$ has at least $k$ accepting runs
    $\rho^{1}\prec\cdots\prec\rho^{k}$ in $\A$.  We can easily construct an
    accepting run $\bar{\rho}$ of $\A_{\geq k}$ on $w$ such that the $\ell$th projection of
    $\bar{\rho}$ is $\rho^{\ell}$ for each $1\leq\ell\leq k$.  We deduce that $\A_{\geq k}$
    accepts exactly the set of words $w\in\Sigma^{+}$ having at least $k$ accepting
    runs in $\A$.
  \end{claimproof}
  
  We deduce from Claim~\ref{claim:bij-run} that if $\A$ is $k$-ambiguous then
  $\A_{\geq k}$ is unambiguous and accepts exactly the words accepted by $\A$ with
  ambiguity $k$.
  
  \begin{claim}\label{claim:k-ambiguous-aperiodic}
    If $\A$ is aperiodic with index $m$ then $\A_{\geq k}$ is aperiodic with index 
    $k(m+1)$.
  \end{claim}
  
  \begin{claimproof}
    Consider a word $w\in\Sigma^{+}$ and a run $\bar{\rho}$ of $\A_{\geq k}$ reading
    $w^{k(m+1)}$.  The sequence of bit vectors along $\bar{\rho}$ is monotone
    component-wise.  Hence, its value can change at most $k-1$ times.  We deduce
    that we can write $\bar{\rho}=\bar{\rho}'\bar{\rho}''\bar{\rho}'''$ where
    $\bar{\rho}''$ reads $w^{m+1}$ with the bit vector unchanged.  Let
    $(\bar{p},\bar{c})$ and $(\bar{q},\bar{c})$ be the source and target states of
    $\rho''$.  The projections $\rho^{1},\ldots,\rho^{k}$ of $\bar{\rho}''$ are
    runs reading $w^{m+1}$ in $\A$.  Since $\A$ is aperiodic with index $m$, we
    find runs $\sigma^{1},\ldots,\sigma^{k}$ reading $w^{m}$ in $\A$ from states
    $p^{1},\ldots,p^{k}$ to $q^{1},\ldots,q^{k}$ respectively.  We may assume that
    for all $1\leq\ell<k$ we have $\sigma^{\ell}=\sigma^{\ell+1}$ if
    $\rho^{\ell}=\rho^{\ell+1}$. Consider the run $\bar{\sigma}$ of $\A_{\geq k}$ 
    starting from $(\bar{p},\bar{c})$ whose projections are 
    $\sigma^{1},\ldots,\sigma^{k}$. It reaches a state $(\bar{q}',\bar{c}')$. 
    Clearly, we have $\bar{q}'=\bar{q}$. We show that $\bar{c}'=\bar{c}$. Let 
    $1\leq\ell<k$. If $c^{\ell}=1$ then $c'^\ell=1$ by definition of $\A_{\geq k}$. If 
    $c^{\ell}=0$ then $\rho^{\ell}=\rho^{\ell+1}$ by definition of $\A_{\geq k}$. We 
    deduce that $\sigma^{\ell}=\sigma^{\ell+1}$ and $c'^{\ell}=0$. Finally, we 
    conclude that $\bar{\rho}'\bar{\sigma}\bar{\rho}'''$ is a run of $\A_{\geq k}$ 
    reading $w^{k(m+1)-1}$ with the same source (resp.\ target) state as 
    $\bar{\rho}$.
    
    By choosing runs $\sigma^{1},\ldots,\sigma^{k}$ reading the word $w^{m+2}$
    instead of $w^{m}$, we otain a run of $\A_{\geq k}$ reading $w^{k(m+1)+1}$ with the
    same source (resp.\ target) state as $\bar{\rho}$.
  \end{claimproof}
  
  Now, let $\A_{\leq k}$ be the minimal automaton for the complement of the 
  language accepted by $\A_{\geq k+1}$. Notice that $\A_{\leq k}$ is 
  deterministic, complete. Moreover, it is aperiodic if $\A$ is aperiodic.
  
  For each $1\leq\ell\leq k$, define the weighted automaton $\A_{\geq 
  k}^{\ell}=(\A_{\geq k},\wgt^{\ell})$ where the weight function corresponds to 
  the $\ell$th path computed by $\A_{\geq k}$. More precisely, we set 
  $\wgt^{\ell}((\bar{q},\bar{c}),a,(\bar{q}',\bar{c}'))=\wgt(q^{\ell},a,q'^{\ell})$.
  Finally, let $\A_k^{\ell}=\A_{\leq k}\times\A_{\geq k}^{\ell}$. It is not 
  difficult to see that $\A_k^{\ell}$ has the following properties.
  
  \begin{claim}\label{claim:Akl}
    The automaton $\A_k^{\ell}$ is unambiguous. A word $w\in\Sigma^{+}$ is 
    in the support of $\A_k^{\ell}$ iff it admits exactly $k$ accepting runs 
    $\rho^{1}\prec\ldots\prec\rho^{k}$ in $\A$. Moreover, in this case, 
    $\usem{\A_k^{\ell}}(w)=\multiset{\wgt(\rho^{\ell}}$. Also, if $\A$ is 
    aperiodic then so is $\A_k^{\ell}$.
  \end{claim}
  
  Finally, to conclude the proof of Theorem~\ref{thm:finitely-ambiguous}, we 
  define for each $1\leq\ell\leq K$ the weighted automaton 
  $\B_\ell=\A_\ell^{\ell}\uplus\cdots\uplus\A_K^{\ell}$. Since the automata 
  $(\A_k^{\ell})_{\ell\leq k\leq K}$ have pairwise disjoint supports, we deduce 
  that $\B_\ell$ is unambiguous. Moreover, using Claim~\ref{claim:Akl} we can 
  easily show that $\usem{\A}=\usem{\B_1\uplus\cdots\uplus\B_K}$.
\end{proof}

\section{Weighted First-Order Logic}\label{sec:wFO}

In this section, we define the syntax and semantics of our weighted first-order
logic.  In \cite{DrosteG-ICALP2005,DrosteG07}, weighted MSO used the
classical syntax of MSO logic; only the semantics over a semiring was changed to
use sums for disjunction and existential quantifications, and products for
conjunctions and universal quantifications.  The possibility to express boolean
properties in \wMSO was obtained via so-called unambiguous formulae.  To improve
readability, a more structured syntax was later used 
\cite{BolligG-DLT09,DrosteM12,KreutzerR2013}, separating a boolean MSO layer with classical boolean
semantics from the higher level of weighted formulas using \emph{products}
($\Prod{X}$, $\Prod{x}$ corresponding to $\forall X$, $\forall x$) and sums
($\Sum{X}$, $\Sum{x}$ corresponding to $\exists X$, $\exists x$) with
quantitative semantics.  As shown in \cite{DrosteG-ICALP2005,DrosteG07}, in
general, to retain equivalence with weighted automata, \wMSO has to be
restricted.  Products $\Prod{X}$ over set variables are disallowed, and
first-order products $\Prod{x}$ must be restricted to finitely valued series 
where the pre-image of each value is recognizable.
This basically means that first-order products cannot be nested or applied after
first-order or second-order sums $\Sum{x}$ or $\Sum{X}$.  This motivated the
equivalent and even more structured syntax of \coreMSO introduced in
\cite{GastinM18}.

As in Section~\ref{sec:wA}, we consider a set $\Weights$ of weights.
The syntax of \wFO is obtained from \coreMSO 
by removing set variables, set quantifications and set sums.  In addition to the
classical boolean first-order logic \eqref{eq:fo}, it has two weighted layers.
Step formulas defined in \eqref{eq:stepfo} consist of constants and if-then-else
applications, where the conditions are formulated in boolean first-order logic.
Finally, \wFO
builds on this by performing products of
step formulas and then applying if-then-else, finite sums, or existential sums.
  \begin{align*}
    \phi &::= \True \mid P_a(x) \mid x\leq y 
    \mid \lnot \phi \mid \phi\land\phi \mid \forall x \phi
    \tag{\FO}\label{eq:fo}
    \\
    \Psi &::= r \mid \Ifthenelse\phi\Psi\Psi 
    \tag{\stepFO}\label{eq:stepfo}
    \\
    \Phi &::= \zero \mid \Prod x \Psi \mid \Ifthenelse\phi\Phi\Phi 
    \mid \Phi+\Phi
    \mid \Sum x \Phi
    \tag{\wFO}\label{eq:wfo}
  \end{align*}
  with $a\in\alphabet$, $r\in\Weights$ and $x,y$ first-order variables.

The semantics of \stepFO formulas is defined inductively. As above, let 
$u\in\Sigma^+$ be a nonempty word and
$\sigma\colon\Variables\to\pos{u}=\{1,\ldots,|u|\}$ be a valuation. For 
\stepFO formulas whose free variables are contained in $\Variables$, we define 
the $\Variables$-semantics as
\begin{align*}
  \stepsem{r}_\Variables(u,\sigma) & =\stepmultiset{r} &  
  \stepsem{\Ifthenelse{\varphi}{\Psi_1}{\Psi_2}}_\Variables(u,\sigma) & =
  \begin{cases}
    \stepsem{\Psi_1}_\Variables(u,\sigma) & \text{if } u,\sigma\models\varphi \\
    \stepsem{\Psi_2}_\Variables(u,\sigma) & \text{otherwise.}
  \end{cases}
\end{align*}
Notice that the semantics of a \stepFO formula is always a single weight from $\Weights$.

For \wFO formulas $\Phi$ whose free variables are contained in $\Variables$,
we define the $\Variables$-semantics
$\usem{\Phi}_\Variables\colon\Sigma_\Variables^+\to\Multiset{\Weights^+}$.
First, we let $\usem{\Phi}_\Variables(\bar{u})=\emptymultiset$ be the empty
multiset when $\bar{u}\in\Sigma_\Variables^+$ is not a valid encoding of a pair
$(u,\sigma)$.  Assume now that $\bar{u}=(u,\sigma)$ is a valid encoding of a
nonempty word $u\in\Sigma^+$ and a valuation
$\sigma\colon\Variables\to\pos{u}$.
The semantics of \wFO formulas is also defined inductively:
$\usem{\zero}_\Variables(u,\sigma)=\emptymultiset$ is the empty multiset, and
\begin{align*}
  \usem{\Prod{x}\Psi}_\Variables(u,\sigma) & = \multiset{r_1r_2\cdots r_{|u|}}
  \text{ where } r_i=\stepsem{\Psi}_{\Variables\cup\{x\}}(u,\sigma[x\mapsto i])
  \text{ for } 1\leq i\leq |u|
  \\
  \usem{\Ifthenelse{\varphi}{\Phi_1}{\Phi_2}}_\Variables(u,\sigma) & =
  \begin{cases}
    \usem{\Phi_1}_\Variables(u,\sigma) & \text{if } u,\sigma\models\varphi \\
    \usem{\Phi_2}_\Variables(u,\sigma) & \text{otherwise}
  \end{cases}
  \\
  \usem{\Phi_1+\Phi_2}_\Variables(u,\sigma) & =
  \usem{\Phi_1}_\Variables(u,\sigma) \uplus \usem{\Phi_2}_\Variables(u,\sigma)
  \\
  \usem{\Sum{x}\Phi}_\Variables(u,\sigma) & =
  \biguplus_{i\in\pos{u}} \usem{\Phi}_{\Variables\cup\{x\}}(u,\sigma[x\mapsto i])
  \,.
\end{align*}
The semantics of the product (first line), is a singleton multiset which
consists of a weight sequence whose length is $|u|$.
We deduce that all weight sequences in a multiset
$\usem{\Phi}_\Variables(u,\sigma)$ have the same length 
and $\usem{\Phi}_\Variables(u,\sigma)\in\Multiset{\Weights^{|u|}}$.
{We simply write $\stepsem{\Psi}$ and $\usem{\Phi}$ when the set
$\Variables$ of variables is clear from the context.}

As explained in Section~\ref{sec:wA}, applying an aggregation function allows to
recover the semantics $\sem{\Phi}$ over semirings such as $\N_{+,\times}$,
$\N_{\max,+}$, etc.  For instance, consider the function
$f\colon\{a,b\}^{+}\to\N$ which assign to a word $w\in\{a,b\}^{+}$ the length of
the maximal $a$-block, i.e., $f(w)=n$ if $a^{n}$ is a factor of $w$ but
$a^{n+1}$ is not.  Over $\N_{\max,+}$, we have $f=\sem{\Phi}$ where
$$
\Phi=\Sum{y,z}\Ifthenelse{(\forall u\, (y\leq u\leq z)\rightarrow P_a(u))}
{(\Prod{x}\Ifthenelse{(y\leq x\leq z)}{1}{0})}{(\Prod{x}0)} \,.
$$
We refer to \cite{DrosteG07} for further examples of quantitative specifications
in weighted logic.

\section{From Weighted Automata to Weighted FO}\label{sec:wa2wfo}
\begin{gpicture}[name=WA2,ignore]
  \node[Nmarks=i,iangle=90](1)(0,0){1}
  \node[Nmarks=i,iangle=90](2)(30,0){2}
  \node[Nmarks=f,fangle=0](3)(15,-12){3}
  \drawloop[loopangle=180](1){$\begin{array}{c} a \mid 2 \\ b \mid 1 \end{array}$}
  \drawloop[loopangle=0](2){$\begin{array}{c} a \mid 3 \\ c \mid 1 \end{array}$}
  \drawedge[curvedepth=1](1,2){$b \mid 1$}
  \drawedge[curvedepth=1](2,1){$c \mid 1$}
  \drawedge[curvedepth=0,ELside=r](1,3){$b \mid 1$}
  \drawedge[curvedepth=0](2,3){$c \mid 1$}
\end{gpicture}
We say that a non-deterministic automaton $\A=(Q,\Sigma,\Delta)$ is
\emph{unambiguous from state $p$ to state $q$} if for all words $u\in\Sigma^+$,
there is at most one run of $\A$ from $p$ to $q$ with label $u$.

\begin{theorem}\label{thm:main1}
  Let $\A$ be an aperiodic weighted automaton which is unambiguous from 
  $p$ to $q$.  We can construct a \wFO sentence
  $\Phi_{p,q}=\Ifthenelse{\varphi_{p,q}}{\Prod{x}\Psi_{p,q}}\zero$ where
  $\varphi_{p,q}$ is a first-order sentence and $\Psi_{p,q}(x)$ is a \stepFO
  formula with a single free variable $x$ such that
  $\usem{\A_{p,q}}=\usem{\Phi_{p,q}}$. 
\end{theorem}

\begin{figure}[tbp]
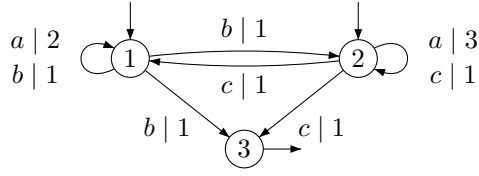

  \centerline{\gusepicture{WA2}}
  \caption{A weighted automaton, which is both aperiodic and unambiguous.}
  \protect\label{fig:WA-unambiguous}
\end{figure}

Before proving Theorem~\ref{thm:main1}, we start with an example.  The automaton
$\A$ of Figure~\ref{fig:WA-unambiguous} is unambiguous and it accepts the
language $\Lang{\A}=(a^*b+a^*c)^+=(a+b+c)^*(b+c)$.
We define a \wFO sentence
$\Phi_{1,3}=\Ifthenelse{\varphi_{1,3}}{\Prod{x}\Psi_{1,3}(x)}\zero$ as follows.
The \FO sentence $\varphi_{1,3}$ checks that $\A$ has a run from state $1$ to 
state $3$ on the input word $w$, i.e., that $w\in a^*b(a^*b+a^*c)^*$:
\begin{align*}
  \varphi_{1,3} &= \exists y\,(P_b(y) \wedge \forall z\,(z<y\implies P_a(z)))
  \wedge \exists y\,(\neg P_a(y) \wedge \forall z\,(z\leq y))
\end{align*}
When this is the case, the \stepFO formula $\Psi_{1,3}(x)$ computes the weight
of the transition taken at a position $x$ in the input word:
\begin{align*}
  \Psi_{1,3}(x) = & \Ifthenelse{(P_b(x) \vee P_c(x))}{1}{} 
  \Ifthenelse{\exists y\,(x<y \wedge P_b(y) \wedge \forall z\,(x<z<y \implies P_a(z)))}{2}{3} \,.
\end{align*}
Notice that the same formula $\Psi=\Psi_{2,3}=\Psi_{1,3}$ also allows to compute the
sequence of weights for the accepting runs starting in state 2.  Therefore, $\A$
is equivalent to the \wFO sentence
\begin{align*}
  \Phi & = \Ifthenelse{\exists y\,(\neg P_a(y) \wedge \forall z\,(z\leq y))}{\Prod{x}\Psi(x)}\zero \,.
\end{align*}

\begin{proof}[Proof of Theorem~\ref{thm:main1}]
  Let $\A=(Q,\Sigma,\Delta,\wgt)$ be the aperiodic weighted automaton.  By
  Theorem~\ref{thm:ap2fo}, for every pair of states $r,s\in Q$ there is a
  first-order sentence $\varphi_{r,s}$ such that
  $\Lang{\A_{r,s}}=\Lang{\varphi_{r,s}}$. This gives in particular the 
  first-order sentence $\varphi_{p,q}$ which is used in $\Phi_{p,q}$.
    
  \begin{claim}\label{claim:1}
    We can construct a \stepFO formula $\Psi_{p,q}(x)$ such that for each word
    $u\in\Lang{\A_{p,q}}$ and each position $1\leq i\leq|u|$ in the word $u$, we
    have $\stepsem{\Psi_{p,q}}(u,x\mapsto i)=\stepmultiset{\wgt(\delta)}$ where
    $\delta$ is the $i$th transition of the unique run $\rho$ of $\A$ from $p$
    to $q$ with label $u$.
  \end{claim}
  
  Before proving this claim, let us show how we can deduce the statement of 
  Theorem~\ref{thm:main1}. 
  Clearly, if a word $u\in\Sigma^+$ is not in $\Lang{\A_{p,q}}$ then we have 
  $\usem{\A_{p,q}}(u)=\emptymultiset=\usem{\Phi_{p,q}}(u)$.
  Consider now a word $u=a_1a_2\cdots a_n\in\Lang{\A_{p,q}}$ and the unique run
  $\rho=\delta_1\delta_2\cdots\delta_n$ of $\A$ from $p$ to $q$ with label $u$.
  We have 
  $\usem{\A_{p,q}}(u)
  =\multiset{\wgt(\delta_1)\wgt(\delta_2)\cdots\wgt(\delta_n)}
  =\usem{\Prod{x}\Psi_{p,q}}(u)$
  where the second equality follows from Claim~\ref{claim:1}.
  We deduce that $\usem{\A_{p,q}}=\usem{\Phi_{p,q}}$. 
  
  \smallskip %
  We turn now to the proof of Claim~\ref{claim:1}.
  Let $\delta=(r,a,s)\in\Delta$ be a transition of $\A$. We define the 
  \FO-formula with one free variable
  $
  \varphi_\delta(x)=\varphi_{p,r}^{<x}\wedge P_a(x)
  \wedge\varphi_{s,q}^{>x}\,.
  $
  \begin{claim}\label{claim:2}
    For each word $u\in\Sigma^+$ and position $1\leq i\leq|u|$, we have 
    $u,x\mapsto i\models\varphi_\delta$ iff $u\in\Lang{\A_{p,q}}$ and $\delta$
    is the $i$th transition of the unique run of $\A$ from $p$ to $q$ with label
    $u$.
  \end{claim}
  
  Indeed, assume that $u,x\mapsto i\models\varphi_\delta$.  Then,
  $u[1,i-1]\models\varphi_{p,r}$ and there is a run $\rho'$ of $\A$ from $p$ to
  $r$ with label $u[1,i-1]$.  Notice that if $i=1$ then $p=r$ and $\rho'$ is 
  the empty run. Similarly, from $u[i+1,|u|]\models\varphi_{s,q}$
  we deduce that there is a run $\rho''$ of $\A$ from $s$ to $q$ with label
  $u[i+1,|u|]$.  Finally, $u,x\mapsto i\models P_a(x)$ means that the $i$th
  letter of $u$ is $a$.  We deduce that $\rho=\rho'\delta\rho''$ is a run of
  $\A$ from $p$ to $q$ with label $u$, hence $u\in\Lang{\A_{p,q}}$.  Moreover,
  $\rho$ is the unique such run since $\A$ is unambiguous from $p$ to $q$. Now, 
  $\delta$ is the $i$th transition of $\rho$, which concludes one direction of 
  the proof.
  Conversely, assume that $u\in\Lang{\A_{p,q}}$ and $\delta$ is the $i$th
  transition of the unique run of $\A$ from $p$ to $q$ with label $u$.  Then,
  $u[1,i-1]\models\varphi_{p,r}$, $u[i+1,|u|]\models\varphi_{s,q}$, and the 
  $i$th letter of $u$ is $a$. Therefore, $u,x\mapsto i\models\varphi_\delta$.
  This concludes the proof of Claim~\ref{claim:2}.
  
  \smallskip %
  Now, choose an arbitrary enumeration $\delta^1,\delta^2,\ldots,\delta^k$ of 
  the transitions in $\Delta$ and define the \stepFO formula with one free 
  variable
  \[
  \Psi_{p,q}(x) = 
  \Ifthenelse{\varphi_{\delta^1}(x)}{\wgt(\delta^1)}{
  \Ifthenelse{\varphi_{\delta^2}(x)}{\wgt(\delta^2)}{
  ~\cdots~
  \Ifthenelse{\varphi_{\delta^k}(x)}{\wgt(\delta^k)}{
  \wgt(\delta^k)}}} \,.
  \]
  We show that this formula satisfies the property of Claim~\ref{claim:1}.
  Consider a word $u\in\Lang{\A_{p,q}}$ and a position $1\leq i\leq|u|$.  Let
  $\delta$ be the $i$th transition of the unique run of $\A$ from $p$ to $q$
  with label $u$.  By Claim~\ref{claim:2}, we have $u,x\mapsto
  i\models\varphi_{\delta^j}$ iff $\delta^j=\delta$.  Therefore,
  $\stepsem{\Psi_{p,q}}(u,x\mapsto i)=\stepmultiset{\wgt(\delta)}$, which concludes the
  proof of Claim~\ref{claim:1}.
\end{proof}

\begin{corollary}\label{cor:aperiodic-unambiguous-to-logic} \mbox{ }
  \begin{enumerate}
    \item Let $\A$ be an aperiodic and unambiguous weighted automaton.  We can
    construct a \wFO sentence $\Phi$ which does not use any $\Sum{x}$ operator
    or $+$ operator, and such that $\usem{\A}=\usem{\Phi}$.
  
    \item Let $\A$ be an aperiodic and finitely ambiguous weighted automaton.
    We can construct a \wFO sentence $\Phi$ which does not use any $\Sum{x}$
    operator, and such that $\usem{\A}=\usem{\Phi}$.
  \end{enumerate}
\end{corollary}

\begin{proof}
  1. Since $\A$ is unambiguous, it is also unambiguous from $p$ to $q$ for all
  $p\in I$ and $q\in F$.  Therefore, a first attempt is the formula
  $\Phi'=\Sum{p\in I,q\in F}\Phi_{p,q}$ where the \wFO sentences $\Phi_{p,q}$
  are given by Theorem~\ref{thm:main1}.  We have $\usem{\A}=\usem{\Phi'}$ and
  the formula $\Phi'$ does not use any $\Sum{x}$ operator, but it does use some
  $+$ operator.  One should notice that, since $\A$ is \emph{unambiguous},
  for any word $u\in\Sigma^{+}$ at most one of the $(\usem{\Phi_{p,q}}(u))_{p\in
  I,q\in F}$ is nonempty. Therefore, if $(p_1,q_1),(p_2,q_2),\ldots,(p_m,q_m)$ 
  is an enumeration of $I\times F$ then we define
  \[
  \Phi = 
  \Ifthenelse{\varphi_{p_1,q_1}}{\Phi_{p_1,q_1}}{
  \Ifthenelse{\varphi_{p_2,q_2}}{\Phi_{p_2,q_2}}{
  ~\cdots~
  \Ifthenelse{\varphi_{p_m,q_m}}{\Phi_{p_m,q_m}}{
  \zero}}} \,.
  \]
  We have $\usem{\A}=\usem{\Phi}$ and
  the formula $\Phi$ does not use any $\Sum{x}$ or $+$ operator. Notice that in 
  the formula above, we may replace $\Phi_{p_i,q_i}$ by 
  $\Prod{x}\Psi_{p_i,q_i}(x)$ as given by Theorem~\ref{thm:main1}.

  Alternatively, by a standard construction adding a new initial and a new final
  state and appropriate transitions, we can obtain an aperiodic weighted automaton
  $\A'$ with a single initial and a single final state such that $\usem{\A'}=\usem\A$
  and, moreover, $\A'$ becomes unambiguous because $\A$ is unambiguous. Then apply
  Theorem~\ref{thm:main1} to $\A'$.
  
  \medskip\noindent 2.\
  Immediate by Theorem~\ref{thm:finitely-ambiguous} and part 1 above.
\end{proof}

Let $\A=(Q,\Sigma,\Delta)$ be a non-deterministic automaton.
Two states $p,q\in Q$ \emph{are in the same strongly connected component} (SCC),
denoted $p\approx q$, if $p=q$ or there exist a run of \A from $p$ to $q$ and
also a run of \A from $q$ to $p$.  Notice that $\approx$ is an equivalence
relation on $Q$.  We denote by $[p]$ the strongly connected component of state
$p$, i.e., the equivalence class of $p$ under $\approx$.  

The automaton $\A$ is \emph{SCC-unambiguous} if it is unambiguous on each
strongly connected component, i.e., $\A$ is unambiguous from $p$ to $q$ for all
$p,q$ such that $p\approx q$.  Notice that a trimmed (all states are reachable
and co-reachable) and unambiguous automaton is SCC-unambiguous.

For instance, the automaton $\A$ of Figure~\ref{fig:WA-SCC-unambiguous} has
three strongly connected components: $\{1\}$, $\{2,3\}$ and $\{4\}$.  It is not
unambiguous from $1$ to $4$, but it is SCC-unambiguous.

\begin{proposition}[\cite{reutenauer1977,ibarra1986} and \cite{Weber_1991} Theorem 4.1]
  Let $\A=(Q,\Sigma,\Delta,I,F)$ be a trimmed non-deterministic automaton. Then
  $\A$ is polynomially ambiguous if and only if $\A$ is SCC-unambiguous.
\end{proposition}

\begin{theorem}\label{thm:main2}
  Let $\A$ be an aperiodic weighted automaton which is SCC-unambiguous.  For
  each pair of states $p$ and $q$, we can construct a \wFO sentence
  $\Phi_{p,q}$ such that $\usem{\A_{p,q}}=\usem{\Phi_{p,q}}$.
  Moreover, we can construct a \wFO sentence $\Phi$ such that
  $\usem{\A}=\usem{\Phi}$.
\end{theorem}

Before starting the proof of Theorem~\ref{thm:main2}, we give for the weighted
automaton $\A$ of Figure~\ref{fig:WA-SCC-unambiguous} the equivalent \wFO
formula
$\Phi_{1,4}=\Sum{y_1}\Sum{y_2}\Ifthenelse{\varphi(y_1,y_2)}{\Prod{x}\Psi(x,y_1,y_2)}\zero$
where $\varphi$ and $\Psi$ are defined below.
When reading a word $w\in\dom{\A}$, the automaton makes two non-deterministic
choices corresponding to the positions $y_1$ and $y_2$ at which the transitions
\emph{switching} between the strongly connected components are taken, i.e.,
transition from state 1 to state 2 is taken at position $y_1$, and transition
from state 3 to state 4 is taken at position $y_2$.  Since the automaton is
SCC-unambiguous, given the input word and these two positions, the run is
uniquely determined.  We use the \FO formula $\varphi(y_1,y_2)$ to check that it
is possible to take the switching transitions at positions $y_1$ and $y_2$:
\begin{align*}
  \varphi(y_1,y_2) = y_1<y_2 & 
  \wedge \forall z\,(z\leq y_1\rightarrow P_a(z))
  \wedge P_a(y_1+1)
  \wedge \forall z\,(y_2\leq z\rightarrow P_b(z)) \,.
\end{align*}
When this is the case, the \stepFO formula $\Psi(x,y_1,y_2)$ computes the weight
of the transition taken at a position $x$ in the input word:
\begin{align*}
  \Psi(x,y_1,y_2) =~ & \Ifthenelse{(x<y_1 \vee y_2<x)}{2}{} 
  \Ifthenelse{(x=y_1 \vee x=y_2)}{1}{} 
  \Ifthenelse{P_a(x+1)}{3}{5} \,.
\end{align*}
With these definitions, we obtain $\usem{\A}=\usem{\Phi_{1,4}}$.

\begin{proof}[Proof of Theorem~\ref{thm:main2}]
  Let $\A=(Q,\Sigma,\Delta,\wgt)$ be the aperiodic weighted automaton which is
  SCC-unambiguous.  Let $p,q\in Q$ be a pair of states of $\A$.  Assume first
  that $p\approx q$ are in the same strongly connected component.  Then $\A$ is
  unambiguous from $p$ to $q$ and we obtain the formula $\Phi_{p,q}$ directly by
  Theorem~\ref{thm:main1}.  So we assume below that $p\not\approx q$ are not in
  the same SCC.
  
  Consider a word $u\in\Lang{\A_{p,q}}$.  Let $\rho$ be a run from $p$
  to $q$ with label $u$. This run starts in the SCC of $p$ and ends in the SCC 
  of $q$. So it uses some transitions linking different SCCs.
  More precisely, we can uniquely split the run as 
  $\rho=\rho_0\delta_1\rho_1\delta_2\rho_2\cdots\delta_m\rho_{m}$ with 
  $m\geq1$ such that each subrun $\rho_i$ stays in some SCC and each transition 
  $\delta_i=(p_i,a_i,q_i)$ switches to a different SCC: 
  \begin{equation}
    p\approx p_1\not\approx q_1\approx p_2\not\approx q_2\approx p_3
    \cdots
    \approx p_m\not\approx q_m\approx q \,.
    \label{eq:switching}
  \end{equation}
  This motivates the following definition.  A sequence of \emph{switching}
  transitions from $p$ to $q$ is a tuple
  $\bar{\delta}=(\delta_1,\ldots,\delta_m)$ with $m\geq1$ satisfying
  \eqref{eq:switching}, where $\delta_i=(p_i,a_i,q_i)$ for $1\leq i\leq m$.
  A $\bar{\delta}$-run from $p$ to $q$ is a run from $p$ to $q$ using exactly
  the sequence of switching transitions $\bar{\delta}$, i.e., a run of the form
  $\rho=\rho_0\delta_1\rho_1\cdots\delta_m\rho_{m}$.  Notice
  that each subrun $\rho_i$ must stay in some SCC of $\A$. 
  
  \begin{claim}\label{claim:3}
    For each sequence $\bar{\delta}$ of switching transitions from $p$ to $q$,
    we can construct a \wFO sentence $\Phi_{p,\bar{\delta},q}$ such that for all
    $u\in\Sigma^+$ we have
    \begin{equation}
      \usem{\Phi_{p,\bar{\delta},q}}(u)=\multiset{\wgt(\rho)\mid\rho 
      \text{ is a $\bar{\delta}$-run from $p$ to $q$ 
      with label } u}\,.
      \label{eq:c3}
    \end{equation}
  \end{claim}
  
  \begin{claimproof}
  During the proof of Claim~\ref{claim:3}, we fix the sequence
  $\bar{\delta}=(\delta_1,\ldots,\delta_m)$ of switching transitions from $p$ to
  $q$, with $m\geq1$ and $\delta_i=(p_i,a_i,q_i)$ for $1\leq i\leq m$.  
  
  \medskip %
  By Theorem~\ref{thm:ap2fo}, for every pair of states $r,s\in Q$ there is a
  first-order sentence $\varphi_{r,s}$ such that
  $\Lang{\A_{r,s}}=\Lang{\varphi_{r,s}}$.  We will use these formulas and also 
  their relativizations $\varphi_{r,s}^{<y}$, $\varphi_{r,s}^{(y,z)}$ and 
  $\varphi_{r,s}^{>z}$.
  
  \medskip %
  We define the \FO formula $\varphi$ 
  with free variables $\Variables=\{y_1,\ldots,y_m\}$ by
  \[
  \varphi = y_1<y_2<\cdots<y_m 
  \wedge \bigwedge_{1\leq i\leq m} P_{a_i}(y_i)
  \wedge \varphi_{p,p_1}^{<y_1}
  \wedge \bigwedge_{1\leq i<m} \varphi_{q_i,p_{i+1}}^{(y_i,y_{i+1})}
  \wedge \varphi_{q_m,q}^{>y_m} \,.
  \]
  Now, we fix a word $u\in\Sigma^+$.
  
  \begin{claim}\label{claim:4}
    There is a bijection between the valuations
    $\sigma\colon\Variables\to\pos{u}=\{1,\ldots,|u|\}$ such that $u,\sigma\models\varphi$ and
    the $\bar{\delta}$-runs $\rho$ from $p$ to $q$ with label $u$.
  \end{claim}
  
  \begin{claimproof}
    First, let $\sigma\colon\Variables\to\pos{u}$ be such that
    $u,\sigma\models\varphi$.  We have
    $\sigma(y_1)<\sigma(y_2)<\cdots<\sigma(y_m)$.
    Since $u,\sigma\models\varphi_{p,p_1}^{<y_1}$, there is a (possibly empty) run $\rho_0(\sigma)$
    from $p$ to $p_1$ reading the prefix $u_0=u[1,\sigma(y_{1})-1]$ of $u$.
    Notice that such a run is unique since $p\approx p_1$ and $\A$ is
    SCC-unambiguous.
    Similarly, for all $1\leq i<m$,
    $u,\sigma\models\varphi_{q_i,p_{i+1}}^{(y_i,y_{i+1})}$ implies that there is a
    \emph{unique} run $\rho_i(\sigma)$ from $q_i$ to $p_{i+1}$ reading the factor
    $u_i=u[\sigma(y_i)+1,\sigma(y_{i+1})-1]$ of $u$.
    Also, $u,\sigma\models\varphi_{q_m,q}^{>y_m}$ implies that there is a
    \emph{unique} run $\rho_m(\sigma)$ from $q_m$ to $q$ reading the suffix
    $u_m=u[\sigma(y_m)+1,|u|]$ of $u$.
    Now, since $u,\sigma\models\bigwedge_{1\leq i\leq m} P_{a_i}(y_i)$, we deduce
    that $u=u_0a_1u_1a_2\cdots a_mu_m$ and that
    $\rho(\sigma)=\rho_0(\sigma)\delta_1\rho_1(\sigma)\cdots\delta_m\rho_{m}(\sigma)$
    is a $\bar{\delta}$-run of $\A$ from $p$ to $q$ with label $u$.
    
    Conversely, let $\rho=\rho_0\delta_1\rho_1\cdots\delta_m\rho_{m}$ be a
    $\bar{\delta}$-run of $\A$ from $p$ to $q$ with label $u$.  Define the
    valuation $\sigma\colon\Variables\to\pos{u}$ so that the switching
    transitions $\bar{\delta}$ along this run are taken at positions
    $\sigma(y_1)<\sigma(y_2)<\cdots<\sigma(y_m)$. We can easily check that 
    $u,\sigma\models\varphi$ and that $\rho=\rho(\sigma)$.
    This concludes the proof of Claim~\ref{claim:4}.
  \end{claimproof}
  
  \medskip %
  Let $\delta=(r,a,s)\in\Delta$ be a transition such that 
  $q_i\approx r\approx s\approx p_{i+1}$ for some $1\leq i<m$. Define the \FO 
  formula
  \[
  \varphi_{\delta} = y_i<x<y_{i+1}
  \wedge \varphi_{q_i,r}^{(y_i,x)}
  \wedge P_a(x)
  \wedge \varphi_{s,p_{i+1}}^{(x,y_{i+1})} \,.
  \]
  It is not difficult to see that for all valuations 
  $\sigma\colon\Variables\cup\{x\}\to\pos{u}$ we have
  $u,\sigma\models\varphi_{\delta}$ iff
  the factor $v=u[\sigma(y_i)+1,\sigma(y_{i+1})-1]$ of $u$ is such that
  $v\in\Lang{\A_{q_i,p_{i+1}}}$ and the unique run of $\A$ from $q_i$ to
  $p_{i+1}$ with label $v$ takes transition $\delta$ on position 
  $\sigma(x)-\sigma(y_i)$. This is similar to Claim~\ref{claim:2}.
  
  Now, if $\delta=(r,a,s)\in\Delta$ is a transition such that 
  $p\approx r\approx s\approx p_{1}$, then we define the \FO 
  formula
  \[
  \varphi_{\delta} = x<y_{1}
  \wedge \varphi_{p,r}^{<x}
  \wedge P_a(x)
  \wedge \varphi_{s,p_{1}}^{(x,y_{1})} \,.
  \]
  Then, $u,\sigma\models\varphi_{\delta}$ iff the prefix
  $v=u[1,\sigma(y_{1})-1]$ of $u$ is such that
  $v\in\Lang{\A_{p,p_{1}}}$ and the unique run of $\A$ from $p$ to
  $p_{1}$ with label $v$ takes transition $\delta$ on position
  $\sigma(x)$.
  
  Next, if $\delta=(r,a,s)\in\Delta$ is a transition such that 
  $q_m\approx r\approx s\approx q$, then we define the \FO 
  formula
  \[
  \varphi_{\delta} = y_m<x
  \wedge \varphi_{q_m,r}^{(y_m,x)}
  \wedge P_a(x)
  \wedge \varphi_{s,q}^{>x} \,.
  \]
  Then, $u,\sigma\models\varphi_{\delta}$ iff the suffix
  $v=u[\sigma(y_{m})+1,|u|]$ of $u$ is such that
  $v\in\Lang{\A_{q_m,q}}$ and the unique run of $\A$ from $q_m$ to
  $q$ with label $v$ takes transition $\delta$ on position
  $\sigma(x)-\sigma(y_m)$.
  
  Finally, for a switching transition $\delta_i$ of $\bar{\delta}$ we let
  $\varphi_{\delta_i}=(x=y_i)$ and for all other transitions
  $\delta=(r,a,s)\in\Delta\setminus\{\delta_1,\ldots,\delta_m\}$ such that $r,s$
  are not both in the strongly connected component of one of the states
  $p_1,p_2,\ldots,p_m,q$ then we let $\varphi_{\delta}=\mathsf{false}$.
  
  As in the proof of Theorem~\ref{thm:main1}, we choose an arbitrary enumeration
  $\delta^1,\delta^2,\ldots,\delta^k$ of the transitions in $\Delta$ and define
  the \stepFO formula with free variables $\Variables\cup\{x\}$
  \[
  \Psi = 
  \Ifthenelse{\varphi_{\delta^1}}{\wgt(\delta^1)}{
  \Ifthenelse{\varphi_{\delta^2}}{\wgt(\delta^2)}{
  ~\cdots~
  \Ifthenelse{\varphi_{\delta^k}}{\wgt(\delta^k)}{
  \wgt(\delta^k)}}} \,.
  \]
  Finally, the \wFO sentence for Claim~\ref{claim:3} is defined by
  \[\textstyle
  \Phi_{p,\bar{\delta},q} = \sum_{y_1}\sum_{y_2}\cdots\sum_{y_m}
  \big( \Ifthenelse{\varphi}{\Prod{x}\Psi}\zero \big) \,.
  \]
  We prove now that Equation~\eqref{eq:c3} holds.
  By definition, $\usem{\Phi_{p,\bar{\delta},q}}(u)$ is the (multiset) union 
  over all valuations $\sigma\colon\Variables\to\pos{u}$
  of $\usem{\Ifthenelse{\varphi}{\Prod{x}\Psi}\zero}(u,\sigma)$.
  By Claim~\ref{claim:4}, there is a bijection between the valuations 
  $\sigma\colon\Variables\to\pos{u}$ such that $u,\sigma\models\varphi$ and the 
  $\bar{\delta}$-runs from $p$ to $q$ with label $u$. Therefore, it remains to 
  show that for all valuations $\sigma\colon\Variables\to\pos{u}$ such that
  $u,\sigma\models\varphi$ with associated $\bar{\delta}$-run $\rho$ we have 
  \[
  \multiset{\wgt(\rho)} = \usem{\Prod{x}\Psi}(u,\sigma) \,.
  \]
  Let $i\in\pos{u}$ and let $\delta$ be the $i$th transition of $\rho$.  From
  the definitions above, we deduce easily that $u,\sigma[x\mapsto i]
  \models\varphi_{\delta^j}$ iff $\delta^j=\delta$.  Therefore,
  $\stepsem{\Psi}(u,\sigma[x\mapsto i])=\stepmultiset{\wgt(\delta)}$.  The announced
  equality $\multiset{\wgt(\rho)} = \usem{\Prod{x}\Psi}(u,\sigma)$ follows.
  This concludes the proof of Claim~\ref{claim:3}.
  \end{claimproof}
    
  To conclude the proof of the first part of Theorem~\ref{thm:main2}, we define
  \[\textstyle
  \Phi_{p,q}=\sum_{\bar{\delta}}\Phi_{p,\bar{\delta},q}
  \]
  where the sum ranges over all sequences $\bar{\delta}$ of switching 
  transitions from $p$ to $q$.
  Recall that we have assumed that $p\not\approx q$ are not in the same SCC of 
  $\A$. Therefore, each run from $p$ to $q$ should go through some sequence of 
  switching transitions. More precisely, given a word $u\in\Sigma^+$, the runs 
  of $\A$ from $p$ to $q$ with label $u$ can be partitionned according to the 
  sequence $\bar{\delta}$ of switching transitions that they use. Therefore, 
  $\usem{\A_{p,q}}(u)$ is the multiset union over all sequences $\bar{\delta}$ 
  of switching transitions from $p$ to $q$ of the multisets
  $\multiset{\wgt(\rho)\mid\rho \text{ is a $\bar{\delta}$-run from $p$ to $q$ 
  with label } u}$. Using Claim~\ref{claim:3}, we deduce that 
  $\usem{\A_{p,q}}(u)=\usem{\Phi_{p,q}}(u)$.
  
  \medskip %
  Finally, consider a weighted automaton with acceptance conditions
  $\A=(Q,\Sigma,\Delta,I,F)$ which is aperiodic and SCC-unambiguous.
  We set $\Phi=\Sum{p\in I,q\in F}\Phi_{p,q}$ where for each pair of states
  $(p,q)\in I\times F$, the formula $\Phi_{p,q}$ is defined as above.
\end{proof}

\section{From Weighted FO to Weighted Automata}\label{sec:wfo2wa}

Let $\A=(Q,\Sigma,\Delta)$ and $\A'=(Q',\Sigma,\Delta')$ be two
non-deterministic automata over the same alphabet $\Sigma$.  Assuming that
$Q\cap Q'=\emptyset$, we define their disjoint union as $\A\uplus\A'=(Q\uplus
Q',\Sigma,\Delta\uplus\Delta')$ and their product as
$\A\times\A'=(Q\times Q',\Sigma,\Delta'')$ where
$\Delta''=\{((p,p'),a,(q,q'))\mid (p,a,p')\in\Delta \text{ and }
(p',a,q')\in\Delta'\}$.

\begin{lemma}\label{lem:union-product-WA}
  The following holds.
  \begin{enumerate}
    \item If $\A$ and $\A'$ are aperiodic, then $\A\uplus\A'$ and $\A\times\A'$ 
    are also aperiodic.
  
    \item If $\A$ and $\A'$ are SCC-unambiguous, then $\A\uplus\A'$ and
    $\A\times\A'$ are also SCC-unambiguous.
  \end{enumerate}
\end{lemma}

Now let $\varphi$ be an \FO-formula with free variables contained in the finite
set $\Variables$, and let $\A_{\varphi,\Variables}=(Q,\Sigma_\Variables,\Delta,\iota,F,G)$ be
the deterministic, complete, trim and aperiodic automaton given by
Theorem~\ref{thm:fo2ap}.
For $i=1,2$, let $\A_i=(Q_i,\Sigma_\Variables,\Delta_i,\wgt_i,I_i,F_i)$ be two
weighted automata over $\Sigma_\Variables$ with $Q_1\cap Q_2=\emptyset$.
We define the weighted automaton $\A'=(Q',\Sigma_\Variables,\Delta',\wgt',I',F')$
by letting
\begin{itemize}
  \item $Q'=Q\times Q_1 \uplus Q\times Q_2$, 
  $I'=\{\iota\}\times I_1 \uplus \{\iota\}\times I_2$,
  $F'=F\times F_1\uplus G\times F_2$,
  
  \item $\Delta'=\{\big((p,p'),a,(q,q')\big) \mid (p,a,q)\in\Delta\text{ and }
  (p',a,q')\in\Delta_1\cup\Delta_2\}$, and 
      
  $\wgt'\big((p,p'),a,(q,q')\big)=\wgt_i(p',a,q')$ if $(p',a,q')\in\Delta_i$ for $i=1,2$.
\end{itemize}
Then we have:

\begin{lemma}\label{lem:tmp8}
  For each $\bar u\in\Sigma_\Variables^+ $, we have
  \[\usem{\A'}(\bar u)=\begin{cases}
  \usem{\A_1}(\bar u), &\text{if }\bar u\text{ is valid and }\bar u\models\varphi,\\
  \usem{\A_2}(\bar u), &\text{if }\bar u\text{ is valid and }\bar u\not\models\varphi,\\
  \emptyset, &\text{if }\bar u\text{ is not valid.}
  \end{cases}\]
  Moreover, if $\A_1$ and $\A_2$ are aperiodic (resp.\ unambiguous,
  SCC-unambiguous) then so is $\A'$.
\end{lemma}

\begin{proof}
  The first part is immediate by the construction of $\A'$ and Theorem~\ref{thm:fo2ap}.
  For the final statement, we can argue as for Lemma~\ref{lem:union-product-WA};
  for the unambiguity part observe that the sets $F$ and $G$ of $\A_{\varphi,\Variables}$
  are disjoint.
\end{proof}

Let $\Variables$ be a finite set of first-order variables and let
$\Variables'=\Variables\cup\{y\}$ where $y\notin\Variables$. 
Given a word $\bar{w}\in\Sigma_{\Variables}^+$ and a position $i\in\pos{w}$, we 
denote by $(\bar{w},y\mapsto i)$ the word over $\Sigma_{\Variables'}$ whose 
projection on $\Sigma_{\Variables}$ is $\bar{w}$ and projection on the 
$y$-component is $0^{i-1}10^{|w|-i}$, i.e., has a unique 1 on position $i$.
Given a function $A\colon\Sigma_{\Variables'}^+\to\Multiset{X}$, we define the
function $\Sum{y}A\colon\Sigma_{\Variables}^+\to\Multiset{X}$ for 
$\bar{w}\in\Sigma_{\Variables}^+$ by
\[
(\Sum{y}A)(\bar{w})=\biguplus_{i\in\pos{w}}A(\bar{w},y\mapsto i) \,.
\]

\begin{lemma}\label{lem:projection}
  Let $\A$ be a weighted automaton over $\Sigma_{\Variables'}$. We can 
  construct a weighted automaton $\A'$ over $\Sigma_{\Variables}$ such that 
  $\usem{\A'}=\Sum{y}\usem{\A}$. Moreover,
  \begin{enumerate}
    \item If $\A$ is aperiodic then $\A'$ is also aperiodic.
  
    \item If $\A$ is SCC-unambiguous then $\A'$ is also SCC-unambiguous.
  \end{enumerate}  
\end{lemma}

\begin{proof}
  Let $\A=(Q,\Sigma_{\Variables'},\Delta,\wgt,I,F)$. We construct 
  $\A'=(Q',\Sigma_{\Variables},\Delta',\wgt',I',F')$ as follows: 
  $Q'=Q\times\{0,1\}$, $I'=I\times\{0\}$, $F'=F\times\{1\}$ and for 
  $\bar{a}\in\Sigma_{\Variables}$ the transitions and weights are given by:
  \begin{itemize}
    \item If $\delta=(p,(\bar{a},0),q)\in\Delta$ then 
    $\delta^0=((p,0),\bar{a},(q,0))\in\Delta'$, 
    $\delta^1=((p,1),\bar{a},(q,1))\in\Delta'$ and
    $\wgt'(\delta^0)=\wgt'(\delta^1)=\wgt(\delta)$.
  
    \item If $\delta=(p,(\bar{a},1),q)\in\Delta$ then 
    $\delta'=((p,0),\bar{a},(q,1))\in\Delta'$ and
    $\wgt'(\delta')=\wgt(\delta)$.
  \end{itemize}
  
  \begin{claim}\label{claim:proj-correct}
    We have $\usem{\A'}=\Sum{y}\usem{\A}$.
  \end{claim}
  
  \begin{claimproof}
    Consider a word $\bar{w}\in\Sigma_{\Variables}^+$ and let $i\in\pos{w}$.  It is
    easy to see that there is a bijection between the accepting runs $\rho$ of
    $\A$ on $(\bar{w},y\mapsto i)$ and the accepting runs $\rho'$ of $\A'$ on
    $\bar{w}$ and switching from $Q\times\{0\}$ to $Q\times\{1\}$ on the $i$th
    transition.  Moreover, this bijection preserves the weight sequences:
    $\wgt'(\rho')=\wgt(\rho)$.  We deduce easily that
    $\usem{\A'}(\bar{w})=(\Sum{y}\usem{\A})(\bar{w})$.
  \end{claimproof}

  \begin{claim}\label{claim:proj-aperiodic}
    If $\A$ is aperiodic then $\A'$ is also aperiodic.
  \end{claim}
  
  \begin{claimproof}
    Assume that $m$ is an aperiodicity index of $\A$.  We claim that $m'=2m$ is an
    aperiodicity index of $\A'$.  Let $\bar{w}\in\Sigma_{\Variables}^+$, let $k\geq
    m'$ and let $\rho'$ be a run of $\A'$ reading $\bar{w}^k$ from some state
    $(p,b)$ to some state $(r,c)$.  We distinguish two cases.  Either there is a
    prefix $\rho'_1$ of $\rho'$ reading $\bar{w}^m$ and staying in $Q\times\{0\}$,
    i.e., $\rho'_1$ goes from $(p,b)=(p,0)$ to some $(q,0)$.  We deduce that there
    is a run $\rho_1$ of $\A$ from $p$ to $q$ and reading $(\bar{w},0)^m$ (recall
    that we denote by $(\bar{w},0)$ the word over $\Sigma_{\Variables'}$ whose
    projection on $\Sigma_{\Variables}$ is $\bar{w}$ and projection on the last
    component belongs to $0^+$).  Since $m$ is an aperiodicity index of $\A$ there
    is another run $\rho_2$ of $\A$ from $p$ to $q$ reading $(\bar{w},0)^{m+1}$.
    We obtain a run $\rho'_2$ of $\A'$ from $(p,0)$ to $(q,0)$ reading
    $\bar{w}^{m+1}$.  Now, replacing the prefix $\rho'_1$ of $\rho'$ by $\rho'_2$
    we obtain a new run $\rho''$ of $\A'$ reading $\bar{w}^{k+1}$ from state
    $(p,0)=(p,b)$ to $(r,c)$.  In the second case, there is a suffix $\rho'_1$ of
    $\rho'$ reading $\bar{w}^m$ from some state $(q,1)$ to $(r,c)=(r,1)$.  We
    construct as above another run $\rho'_2$ from $(q,1)$ to $(r,1)$ reading
    $\bar{w}^{m+1}$.  Replacing the suffix $\rho'_1$ of $\rho'$ by $\rho'_2$, we 
    obtain the run $\rho''$ from $(p,b)$ to $(r,c)$ reading $\bar{w}^{k+1}$. 
    Finally, when $k>m'=2m$, a similar argument allows to construct a run 
    $\rho''$ from $(p,b)$ to $(r,c)$ reading $\bar{w}^{k-1}$.    
  \end{claimproof}
  
  \begin{claim}\label{claim:proj-unambiguous}
    If $\A$ is SCC-unambiguous then $\A'$ is also SCC-unambiguous.
  \end{claim}
  
  \begin{claimproof}
    Let $\bar{w}\in\Sigma_{\Variables}^+$ and let $(p,b)\approx'(q,c)$ be two
    states of $Q'$ which are in the same SCC of $\A'$.  Then, $b=c$ and $p\approx
    q$ are in the same SCC of $\A$.  Since $b=c$, there is a bijection between the
    runs of $\A'$ from $(p,b)$ to $(q,c)$ reading $\bar{w}$ and the runs of $\A$
    from $p$ to $q$ reading $(\bar{w},0)$.  Since $\A$ is SCC-unambiguous and
    $p\approx q$, there is at most one run of $\A$ from $p$ to $q$ reading
    $(\bar{w},0)$.  Hence, there is at most one run of $\A'$ from $(p,b)$
    to $(q,c)$ reading $\bar{w}$.
  \end{claimproof}
\end{proof}

We turn now to one of our main results: given a \stepFO formula $\Psi$, we can
construct a weighted automaton for $\Prod{x}\Psi$ which is both aperiodic and
unambiguous.

When weights are uninterpreted, a weighted automaton
$\A=(Q,\Sigma,\Delta,\wgt,I,F)$ is a letter-to-letter transducer from its input
alphabet $\Sigma$ to the output alphabet $\Weights$.  If in addition the input
automaton is unambiguous, then we have a functional transducer. In the 
following lemma, we will construct such functional transducers using the 
boolean output alphabet $\mathbb{B}=\{0,1\}$.

\begin{lemma}\label{lem:fo2transducer}
  Let $\Variables=\{y_1,\ldots,y_m\}$.  Given an \FO formula $\varphi$ with free
  variables contained in $\Variables'=\Variables\cup\{x\}$, we can construct a
  transducer $\BphiV$ from $\Sigma_\Variables$ to $\mathbb{B}$
  which is aperiodic and unambiguous and such that for all words
  $\bar{w}\in\Sigma_\Variables^+$
  \begin{enumerate}
    \item  there is a (unique)
    accepting run of $\BphiV$ on the input word $\bar{w}$ iff
    it is a valid encoding of a pair $(w,\sigma)$ where
    $w\in\Sigma^+$ and $\sigma\colon\Variables\to\pos{w}$ is a valuation, 
  
    \item and in this case, for all $1\leq i\leq|w|$, the $i$th bit of the
    output is 1 iff $w,\sigma[x\mapsto i]\models\varphi$.
  \end{enumerate}
\end{lemma}

\begin{proof}
  Notice that $\Sigma_{\Variables'}=\Sigma_\Variables\times\mathbb{B}$ so 
  letters in $\Sigma_{\Variables'}$ are of the form $(\bar{a},0)$ or $(\bar{a},1)$ where 
  $\bar{a}\in\Sigma_{\Variables}$.  Abusing the notations, when
  $\bar{v}\in\Sigma_\Variables^*$, we write $(\bar{v},0)$ to denote the word over 
  $\Sigma_{\Variables'}$ whose projection on $\Sigma_\Variables$ is $\bar{v}$ and 
  projection on the $x$-component consists of $0$'s only.

  Consider the deterministic, complete and aperiodic automaton
  $\AphiVp=(Q,\Sigma_{\Variables'},\Delta,\iota,F,G)$ associated
  with $\varphi$ by Theorem~\ref{thm:fo2ap}.  
  We also denote by $\Delta$ the extension of the transition function to subsets
  of $Q$. So we see the deterministic and complete transition relation both as 
  a total function $\Delta\colon Q\times\Sigma_{\Variables'}\to Q$ and
  $\Delta\colon2^Q\times\Sigma_{\Variables'}\to2^Q$.

  We construct now the transducer
  $\BphiV=(Q',\Sigma_\Variables,\Delta',\wgt,I',F')$.  The set
  of states is $Q'=Q\times2^Q\times2^Q\times\mathbb{B}$.  The unique initial
  state is $\iota'=(\iota,\emptyset,\emptyset,0)$.
  The set of final states is
  $F'=(Q\times2^F\times2^G\times\mathbb{B})\setminus\{\iota'\}$.
  Then, we define the following transitions:
  \begin{itemize}
    \item $\delta=((p,X,Y,b),\bar{a},(p',X',Y',1))\in\Delta'$ is a transition with 
    weight $\wgt(\delta)=1$ if \\
    $p'=\Delta(p,(\bar{a},0))$, 
    $X'=\Delta(X,(\bar{a},0))\cup\{\Delta(p,(\bar{a},1))\}$ and
    $Y'=\Delta(Y,(\bar{a},0))$,
  
    \item $\delta=((p,X,Y,b),\bar{a},(p',X',Y',0))\in\Delta'$ is a transition with 
    weight $\wgt(\delta)=0$ if \\
    $p'=\Delta(p,(\bar{a},0))$, 
    $X'=\Delta(X,(\bar{a},0))$ and
    $Y'=\Delta(Y,(\bar{a},0))\cup\{\Delta(p,(\bar{a},1))\}$.
  \end{itemize}
  Notice that, whenever we read a new input letter $\bar{a}\in\Sigma_\Variables$,
  there is a non-deterministic choice.  In the first case above, we guess that
  formula $\varphi$ will hold on the input word when the valuation is extended by
  assigning $x$ to the current position, whereas in the second case we guess
  that $\varphi$ will not hold.  The guess corresponds to the output of the
  transition, as required by the second condition of
  Lemma~\ref{lem:fo2transducer}.  Now, we have to check that the guess is
  correct.  For this, the first component of $\BphiV$ computes
  the state $p=\Delta(\iota,(\bar{u},0))$ reached by $\AphiVp$
  after reading $(\bar{u},0)$ where $\bar{u}\in\Sigma_\Variables^*$ is the
  current prefix of the input word.  When reading the current letter
  $\bar{a}\in\Sigma_\Variables$, the transducer adds the state
  $\Delta(p,(\bar{a},1))=\Delta(\iota,(\bar{u},0)(\bar{a},1))$ either to the
  ``positive'' $X$-component or to the ``negative'' $Y$-component of its state,
  depending on its guess as explained above.  Then, the transducer continues
  reading the suffix $\bar{v}\in\Sigma_\Variables^*$ of the input word.  It
  updates the $X$ (resp.\ $Y$)-component so that it contains the state 
  $q=\Delta(\iota,(\bar{u},0)(\bar{a},1)(\bar{v},0))$ at the end of the run. 
  Now, the acceptance condition allows us to check that the guess was correct.
  \begin{enumerate}
    \item If $\bar{w}=\bar{u}\bar{a}\bar{v}$ is not a valid encoding of a pair 
    $(w,\sigma)$ with $w\in\Sigma^+$ and $\sigma\colon\Variables\to\pos{w}$ 
    then $q\notin F\cup G$ and the run of the transducer is not accepting.
    Otherwise, let $i\in\pos{w}$ be the position where the guess was made.
    
    \item
    If the guess was positive then $q$ belongs to the $X$-component and the
    accepting condition implies $q\in F$, which means by definition of
    $\AphiVp$ that $w,\sigma[x\mapsto i]\models\varphi$.
  
    \item If the guess was negative then $q$ belongs to the $Y$-component and
    the accepting condition implies $q\in G$, which means by definition of
    $\AphiVp$ that $w,\sigma[x\mapsto i]\not\models\varphi$.
  \end{enumerate}
  We continue the proof with several remarks.  
  
  First, since the automaton $\AphiVp$ is
  complete, after reading a nonempty input word $\bar{w}\in\Sigma_\Variables^+$
  the transducer cannot be back in its initial state
  $\iota'=(\iota,\emptyset,\emptyset,0)$. This is because the second and third 
  components of the state cannot both be empty. Since $\iota'\notin F'$, the 
  support of the transducer consists of nonempty words only. 
  
  Second, consider a run of the transducer on some input word
  $\bar{w}\in\Sigma_\Variables^+$ from its initial state $\iota'$ to some state 
  $(p,X,Y,b)$. As explained above, one can check that $X\cup Y\subseteq F\cup 
  G$ iff $\bar{w}$ is a valid encoding of a pair $(w,\sigma)$. Therefore, the 
  support of the transducer consists of valid encodings only.
  
  Now, consider a valid encoding $\bar{w}$ of a pair $(w,\sigma)$ and consider a
  run $\rho$ of $\BphiV$ on $\bar{w}$ from $\iota'$ to some
  state $(p,X,Y,b)$.  This run is entirely determined by the sequence of guesses
  made at every position of the input word.  As explained above, one can check
  that all guesses are correct iff $X\subseteq F$ and $Y\subseteq G$.
  Therefore, $\BphiV$ admits a unique accepting run on
  $\bar{w}$.  This shows that the support of $\BphiV$ is exactly
  the set of valid encodings, that this transducer is unambiguous, and that the
  last condition of the lemma holds, i.e., the $i$th bit of the output is 1 iff
  $w,\sigma[x\mapsto i]\models\varphi$.
  
  \medskip %
  To complete the proof, it remains to show that $\BphiV$ is 
  aperiodic. Let $m\geq1$ be an aperiodicity index of 
  $\AphiVp$. We claim that $m'=2m+2|Q|$ is an aperiodicity index 
  of $\BphiV$.
  Let $\alpha=(p,X,Y,b)$ and $\alpha'=(p',X',Y',b')$ be two states of
  $\BphiV$ and let $\bar{w}\in\Sigma_\Variables^+$ be a
  nonempty word.  
  
  Assume first that there is a run $\rho$ of $\BphiV$ from
  $\alpha$ to $\alpha'$ reading the input word $\bar{w}^{k}$ with $k\geq 2m+1$.
  We show that there is another run of $\BphiV$ from $\alpha$
  to $\alpha'$ reading the input word $\bar{w}^{k+1}$.  We split $\rho$ in three
  parts: $\rho=\rho_1\rho_2\rho_3$ where $\rho_1$ reads the prefix $\bar{w}^m$,
  $\rho_2$ reads $\bar{w}$ and $\rho_3$ reads the suffix $\bar{w}^{k-m-1}$.
  Consider the intermediary states $\alpha_i=(q_i,X_i,Y_i,b_i)$ reached
  after $\rho_i$ ($1\leq i\leq 3$):
  $\alpha\xrightarrow{\rho_1}\alpha_1\xrightarrow{\rho_2}\alpha_2\xrightarrow{\rho_3}\alpha_3=\alpha'$.
  Since $\AphiVp$ is deterministic with aperiodicity index $m$
  we obtain
  $\Delta(p,(\bar{w},0)^m)=\Delta(p,(\bar{w},0)^{m+1})=\Delta(p,(\bar{w},0)^k)$.
  Therefore, $q_1=q_2=q_3=p'$.
  
  Notice that, by definition of the transitions of $\BphiV$, a
  run is entirely determined by its starting state, its input word, and the
  sequence of choices which is indicated in the fourth component of the states.
  Let $\rho'_2$ be the run starting from $\alpha_2$, reading $\bar{w}$ and
  following the same sequence of choices as $\rho_2$.  Let
  $\alpha'_2=(q'_2,X'_2,Y'_2,b'_2)$ be the state reached after $\rho'_2$.  
  Let also $\rho'_3$ be the run starting from $\alpha'_2$, reading
  $\bar{w}^{k-m-1}$ and following the same sequence of choices as $\rho_3$.  Let
  $\alpha'_3=(q'_3,X'_3,Y'_3,b'_3)$ be the state reached after $\rho'_3$.  Thus,
  we obtain a run
  $\rho'=\alpha\xrightarrow{\rho_1}\alpha_1\xrightarrow{\rho_2}\alpha_2
  \xrightarrow{\rho'_2}\alpha'_2\xrightarrow{\rho'_3}\alpha'_3$
  reading the input word $\bar{w}^{k+1}$.  It remains to show that
  $\alpha'_3=\alpha_3$.  As above, we have
  $q'_3=\Delta(p,(\bar{w},0)^{k+1})=\Delta(p,(\bar{w},0)^k)=q_3$.  Also, $b'_3$
  stores the last choice of $\rho'_3$, which is the same as the last choice of
  $\rho_3$ stored in $b_3$ and we get $b'_3=b_3$.  It remains to show that
  $X'_3=X_3$ and $Y'_3=Y_3$.  
  To this end, we introduce yet another variant of the runs $\rho_2$ and
  $\rho_3$.  Let $\rho''_2$ be the run starting from
  $(p',\emptyset,\emptyset,0)$, reading $\bar{w}$ and following the same
  sequence of choices as $\rho_2$.  Let $\alpha''_2=(q''_2,X''_2,Y''_2,b''_2)$
  be the state reached after $\rho''_2$.  It is easy to see that $q''_2=q_2=p'$
  and $b''_2=b_2$.  Moreover, we have 
  \begin{align*}
    X_2 &= X''_2\cup\Delta(X_1,(\bar{w},0)) & 
    X'_2 &= X''_2\cup\Delta(X_2,(\bar{w},0)) 
    \\
    Y_2 &= Y''_2\cup\Delta(Y_1,(\bar{w},0)) &
    Y'_2 &= Y''_2\cup\Delta(Y_2,(\bar{w},0)) \,.
  \end{align*}
  Similarly, let $\rho''_3$ be the run starting from
  $(p',\emptyset,\emptyset,0)$, reading $\bar{w}^{k-m-1}$ and following the same
  sequence of choices as $\rho_3$.  Let $\alpha''_3=(p',X''_3,Y''_3,b_3)$
  be the state reached after $\rho''_3$.  
  We have 
  \begin{align*}
    X_3 &= X''_3\cup\Delta(X_2,(\bar{w},0)^{k-m-1}) & 
    X'_3 &= X''_3\cup\Delta(X'_2,(\bar{w},0)^{k-m-1}) 
    \\
    Y_3 &= Y''_3\cup\Delta(Y_2,(\bar{w},0)^{k-m-1}) &
    Y'_3 &= Y''_3\cup\Delta(Y'_2,(\bar{w},0)^{k-m-1}) \,.
  \end{align*}
  Notice that $k-m-1\geq m$, hence we get
  $\Delta(X_2,(\bar{w},0)^{k-m-1})=\Delta(X_2,(\bar{w},0)^{k-m})$ from the
  aperiodicity of $\AphiVp$.  Finally, using $X''_2\subseteq
  X_2$, we obtain
  $\Delta(X'_2,(\bar{w},0)^{k-m-1})=\Delta(X_2,(\bar{w},0)^{k-m-1})$ and
  $X'_3=X_3$.  Similarly, we prove that $Y'_3=Y_3$.
  
  \medskip%
  Conversely, we assume that there is a run $\rho$ of $\BphiV$ from $\alpha$ to
  $\alpha'$ reading the input word $\bar{w}^{k}$ with $k>m'=2m+2|Q|$.  We show that
  there is another run $\rho'$ of $\BphiV$ from $\alpha$ to $\alpha'$ reading the input
  word $\bar{w}^{k-1}$.  We split $\rho$ in $2|Q|+3$ parts:
  $\rho=\rho_0\rho_1\cdots\rho_{2|Q|+1}\rho_{2|Q|+2}$ where $\rho_0$ reads the
  prefix $\bar{w}^{k-2|Q|-m-1}$, each $\rho_i$ with $1\leq i\leq 2|Q|+1$ reads
  $\bar{w}$, and $\rho_{2|Q|+2}$ reads the suffix $\bar{w}^{m}$.
  Consider the intermediary states $\alpha_i=(q_i,X_i,Y_i,b_i)$ reached after
  $\rho_i$ ($0\leq i\leq 2|Q|+2$). We have
  \[
  \alpha\xrightarrow{\rho_0}\alpha_0\xrightarrow{\rho_1}\alpha_1\cdots
  \alpha_{2|Q|+1}\xrightarrow{\rho_{2|Q|+2}}\alpha_{2|Q|+2}=\alpha' \,.
  \]
  Since $k-2|Q|-m-1\geq m$ and $\AphiVp$ is deterministic with aperiodicity
  index $m$, we deduce that $q_0=q_1=\cdots=q_{2|Q|+1}=q_{2|Q|+2}=p'$.  As in
  the previous part of the aperiodicity proof, for each $1\leq i\leq2|Q|+2$, we
  consider the run $\rho'_i$ starting form $(p',\emptyset,\emptyset,0)$, reading
  the same input word as $\rho_i$ and making the same sequence of choices as
  $\rho_i$.  Let $\alpha'_i=(p',X'_i,Y'_i,b_i)$ be the state reached after
  $\rho'_i$ ($1\leq i\leq 2|Q|+2$).  We have, for all $1\leq i\leq2|Q|+1$:
  \begin{align*}
    X_i &= X'_i\cup\Delta(X_{i-1},(\bar{w},0)) & 
    X_{2|Q|+2} &= X'_{2|Q|+2}\cup\Delta(X_{2|Q|+1},(\bar{w},0)^m) 
    \\
    Y_i &= Y'_i\cup\Delta(Y_{i-1},(\bar{w},0)) &
    Y_{2|Q|+2} &= Y'_{2|Q|+2}\cup\Delta(Y_{2|Q|+1},(\bar{w},0)^m) \,.
  \end{align*}
  The states in $X'=X_{2|Q|+2}$ and $Y'=Y_{2|Q|+2}$ originate from the initial
  sets $X_0$ and $Y_0$ and from the sets $X'_i$ and $Y'_i$ created by the
  subruns $\rho_i$ ($1\leq i\leq 2|Q|+2$).  Intuitively, there is at least one
  index $1\leq i\leq 2|Q|+1$ such that the contribution of $\rho_i$ is subsumed
  by other subruns (formal proof below).  Removing the subrun $\rho_i$
  yields the desired run $\rho'$ of $\BphiV$ from $\alpha$ to $\alpha'$ reading
  the input word $\bar{w}^{k-1}$ (formal proof below).
  
  For $0\leq i\leq 2|Q|+1$, we let $k_i=2|Q|+1-i+m$.  For $1\leq i\leq 2|Q|+2$,
  we define by descending induction on $i$ the contributions $X''_i$ and $Y''_i$
  to $X'=X_{2|Q|+2}$ and $Y'=Y_{2|Q|+2}$ which originate from subruns $\rho_j$
  with $j\geq i$:
  \begin{align*}
    X''_{2|Q|+2} & =X'_{2|Q|+2} & 
    X''_i & = X''_{i+1} \cup \Delta(X'_i,(\bar{w},0)^{k_i}) 
    \\
    Y''_{2|Q|+2} & =Y'_{2|Q|+2} & 
    Y''_i & = Y''_{i+1} \cup \Delta(Y'_i,(\bar{w},0)^{k_i}) \,.
  \end{align*}
  We deduce easily that for all $1\leq i\leq 2|Q|+2$ we have
  \begin{align*}
    X_{2|Q|+2} & = X''_{i}\cup\Delta(X_{i-1},(\bar{w},0)^{k_{i-1}})
    &
    Y_{2|Q|+2} & = Y''_{i}\cup\Delta(Y_{i-1},(\bar{w},0)^{k_{i-1}}) \,.
  \end{align*}
  Let $1\leq i\leq 2|Q|+1$ be such that $X''_i=X''_{i+1}$ and $Y''_i=Y''_{i+1}$.
  Using the monotonicity of the sequences, it is easy to see that such an index $i$
  must exist. We show that we can remove the subrun $\rho_i$. Let $\rho''$ be 
  the run from $\alpha_{i-1}$ (and not $\alpha_i$) which reads $\bar{w}^{k_i}$ 
  and makes the same sequence of choices as $\rho_{i+1}\cdots\rho_{2|Q|+2}$.
  Let $\alpha''=(q'',X'',Y'',b'')$ be the state reached after $\rho''$. It is 
  easy to see that $q''=q_{2|Q|+2}=p'$ and $b''=b_{2|Q|+2}=b'$. We show that 
  $X''=X_{2|Q|+2}=X'$. Since $\rho''$ makes the same sequence of choices as 
  $\rho_{i+1}\cdots\rho_{2|Q|+2}$, we see that the contribution to $X''$ coming 
  from $\rho''$ is exactly $X''_{i+1}$. Therefore, 
  \[
  X''=X''_{i+1}\cup\Delta(X_{i-1},(\bar{w},0)^{k_i})
  =X''_{i}\cup\Delta(X_{i-1},(\bar{w},0)^{k_{i-1}})
  =X_{2|Q|+2}=X' 
  \]
  where the second equality follows from the hypothesis $X''_i=X''_{i+1}$ and
  the aperiodicity of \AphiVp with index $m$ since $k_{i-1}=k_i+1>m$.
  Similarly, we can prove that $Y''=Y'$ and we obtain $\alpha''=\alpha'$.
  Therefore, $\rho'=\rho_0\cdots\rho_{i-1}\rho''$ is the desired run of $\BphiV$
  from $\alpha$ to $\alpha'$ reading the input word $\bar{w}^{k-1}$. This 
  concludes the proof of aperiodicity of \BphiV with index $m'=2|Q|+2m$.
\end{proof}

\begin{theorem}\label{thm:stepwFO2wA}
  Let $\Variables=\{y_1,\ldots,y_m\}$.  Given a \stepFO formula $\Psi$ with free
  variables contained in $\Variables'=\Variables\cup\{x\}$, we can construct a
  weighted automaton $\ApsiV$ over $\Sigma_\Variables$ which is aperiodic and
  unambiguous and which is equivalent to $\Prod{x}\Psi$, i.e., such that
  $\usem{\ApsiV}(\bar{w})=\usem{\Prod{x}\Psi}_\Variables(\bar{w})$ for all words
  $\bar{w}\in\Sigma_\Variables^+$.
\end{theorem}

\begin{proof}
  In case $\Psi=r$ is an atomic \stepFO formula, we replace it with the
  equivalent $\Ifthenelse{\True}{r}r$ \stepFO formula.  Let
  $\varphi_1,\ldots,\varphi_k$ be the \FO formulas occurring in $\Psi$.  By the
  above remark, we have $k\geq1$.
  Consider the aperiodic and unambiguous transducers $\B_1,\ldots,\B_k$ given by
  Lemma~\ref{lem:fo2transducer}.  For $1\leq i\leq k$, we let
  $\B_i=(Q_i,\Sigma_{\Variables},\Delta_i,\wgt_i,I_i,F_i)$.  The weighted
  automaton $\ApsiV=(Q,\Sigma_{\Variables},\Delta,\wgt,I,F)$ is essentially a
  cartesian product of the transducers $\B_i$.  More precisely, we let
  $Q=\prod_{i=1}^{k}Q_i$, $I=\prod_{i=1}^{k}I_i$, $F=\prod_{i=1}^{k}F_i$, and
  \[
  \Delta=\{((p_1,\ldots,p_k),\bar{a},(q_1,\ldots,q_k)) \mid
  (p_i,\bar{a},q_i)\in\Delta_i \text{ for all } 1\leq i\leq k\} \,.
  \]
  Since the transducers $\B_i$ are all aperiodic and unambiguous, we deduce by
  Lemma~\ref{lem:union-product-WA} that $\ApsiV$ is also aperiodic and
  unambiguous.  It remains to define the weight function $\wgt$.
  
  Given a bit vector $\bar{b}=(b_1,\ldots,b_k)\in\mathbb{B}^k$ of size $k$, we
  define $\Psi(\bar{b})$ as the weight from $\Weights$ resulting from the
  \stepFO formula $\Psi$ when the \FO conditions $\varphi_1,\ldots,\varphi_k$
  evaluate to $\bar{b}$.  Formally, the definition is by structural induction on
  the \stepFO formula:
  \begin{align*}
    r(\bar{b}) & = r & 
    (\Ifthenelse{\varphi_i}{\Psi_1}{\Psi_2})(\bar{b}) & =
    \begin{cases}
      \Psi_1(\bar{b}) & \text{if } b_i=1 \\
      \Psi_2(\bar{b}) & \text{if } b_i=0 \,.
    \end{cases}
  \end{align*}
  Consider a transition
  $\delta=((p_1,\ldots,p_k),\bar{a},(q_1,\ldots,q_k))\in\Delta$ and let
  $\delta_i=(p_i,\bar{a},q_i)$ for $1\leq i\leq k$.  Let
  $\bar{b}=(b_1,\ldots,b_k)\in\mathbb{B}^k$ where
  $b_i=\wgt(\delta_i)\in\mathbb{B}$ for all $1\leq i\leq k$. We define
  $\wgt(\delta)=\Psi(\bar{b})$.
  
  Let $\bar{w}\in\Sigma_\Variables^+$.  If $\bar{w}$ is not a valid encoding of
  a pair $(w,\sigma)$ then
  $\usem{\Prod{x}\Psi}_\Variables(\bar{w})=\emptymultiset$ by definition.
  Moreover, $\usem{\ApsiV}(\bar{w})=\emptymultiset$ since by
  Lemma~\ref{lem:fo2transducer}, $\bar{w}$ is not in the support of $\B_1$.  We
  assume below that $\bar{w}$ is a valid encoding of a pair $(w,\sigma)$ where
  $w\in\Sigma^+$ and $\sigma\colon\Variables\to\pos{w}$ is a valuation.  Then,
  each transducer $\B_i$ admits a unique accepting run $\rho_i$ reading the
  input word $\bar{w}$.  These result in the unique accepting run $\rho$ of
  \ApsiV reading $\bar{w}$.  The projections of $\rho$ on $\B_1,\ldots,\B_k$ are
  $\rho_1,\ldots,\rho_k$.  Let $j\in\pos{w}=\{1,\ldots,|w|\}$ be a position in
  $\bar{w}$ and let $\delta^{j}$ be the $j$-th transition of $\rho$.  For $1\leq
  i\leq k$, we denote by $\delta_i^{j}$ the projection of $\delta^{j}$ on $\B_i$ and we
  let $b_i^{j}=\wgt(\delta_i^{j})$.  By Lemma~\ref{lem:fo2transducer}, we get
  $b_i^{j}=1$ iff $w,\sigma[x\mapsto j]\models\varphi_i$.
  Finally, let $\bar{b}^{j}=(b_1^{j},\ldots,b_k^{j})$. From the above, we deduce that 
  $\stepsem{\Psi}_{\Variables\cup\{x\}}(w,\sigma[x\mapsto 
  j])=\stepmultiset{\Psi(\bar{b}^{j})}=\stepmultiset{\wgt(\delta^{j})}$.
  Putting things together, we have 
  \begin{align*}
    \usem{\ApsiV}(w,\sigma)=\multiset{\wgt(\rho)}
    &= \multiset{\wgt(\delta^{1})\cdots\wgt(\delta^{|w|}}
    = \usem{\Prod{x}\Psi}_\Variables(w,\sigma) \,.
    \qedhere
  \end{align*}
\end{proof}

\begin{theorem}\label{thm:tmp12}
  Let $\Phi$ be a \wFO sentence. We can construct an aperiodic SCC-unambiguous
  weighted automaton $\A$ such that $\usem\A=\usem\Phi$. Moreover, if $\Phi$ does
  not contain the sum operations $+$ and $\Sum{x}$, then $\A$ can be chosen to be
  unambiguous. If $\Phi$ does not contain the sum operation $\Sum{x}$, we can
  construct $\A$ as a finite union of unambiguous weighted automata.
\end{theorem}

\begin{proof}
  We proceed by structural induction on $\Phi$.
  For $\Phi=\zero$ this is trivial. For $\Phi=\Prod{x}\Psi$ with a \stepFO formula
  $\Psi$, we obtain an aperiodic unambiguous weighted automaton $\A$ by
  Theorem~\ref{thm:stepwFO2wA}.
  For formulas $\Ifthenelse\varphi{\Phi_1}{\Phi_2}$, $\Phi_1+\Phi_2$ and $\Sum{x}\Phi$, we apply
  Lemmas~\ref{lem:tmp8}, \ref{lem:union-product-WA} and \ref{lem:projection},
  respectively.
\end{proof}

In the proof of Theorem~\ref{thm:tmp12}, we may obtain the final statement also as
a consequence of the preceding one by the following observations which could be of
independent interest.
Let $\varphi$ be an \FO-formula and $\Phi_1$, $\Phi_2$ two \wFO formulas, each
with free variables contained in $\Variables$. Then,
\begin{align*}
  &\usem{\Ifthenelse{\varphi}{\Phi_1}{\Phi_2}}_\Variables =
      \usem{\Ifthenelse{\varphi}{\Phi_1}{\zero}+
        \lnot\Ifthenelse{\varphi}{\Phi_2}{\zero}}_\Variables\;,\\
  &\usem{\Ifthenelse{\varphi}{\Phi_1+\Phi_2}\zero}_\Variables =
      \usem{\Ifthenelse{\varphi}{\Phi_1}{\zero}+
        \Ifthenelse{\varphi}{\Phi_2}{\zero}}_\Variables\;.
\end{align*}
Hence, given a \wFO sentence $\Phi$ not containing the sum operation $\Sum{x}$,
we can rewrite $\Phi$ as a sum of $\zero$, $\Prod{x}\Psi$ and if-then-else sentences
of the form $\Ifthenelse\varphi{\Phi'}\zero$ where $\Phi'$ does not contain the sum
operations $+$ or $\Sum{x}$.

\begin{proof}[Proof of Thm~\ref{thm:main}]
  Immediate by Theorem~\ref{thm:main2}, Theorem~\ref{thm:finitely-ambiguous},
  Corollary~\ref{cor:aperiodic-unambiguous-to-logic} and 
  Theorem~\ref{thm:tmp12}.
\end{proof}

\section{Examples}\label{sec:examples}
\begin{gpicture}[name=A3,ignore]%
  \node[Nmarks=if,iangle=-90,fangle=90](1)(0,0){}
  \node[Nmarks=if,iangle=-90,fangle=90](2)(20,0){}
  \drawloop[loopangle=180,ELdist=-1](1){$\begin{array}{c} a \mid 1 \\ b \mid 0 \\ c \mid 0 \end{array}$}
  \drawloop[loopangle=0,ELdist=-1](2){$\begin{array}{c} a \mid 0 \\ b \mid 1 \\ c \mid 0 \end{array}$}
  \drawedge[curvedepth=2](1,2){$c \mid 0$}
  \drawedge[curvedepth=2](2,1){$c \mid 0$}
\end{gpicture}%
\begin{gpicture}[name=A,ignore]
  \node[Nmarks=i](1)(40,0){}
  \node[Nmarks=f](2)(60,0){}
  \drawloop(1){$a \mid 1$}
  \drawedge[curvedepth=2](1,2){$a \mid 1$}
  \drawedge[curvedepth=2](2,1){$a \mid 1$}
\end{gpicture}
\begin{gpicture}[name=A5,ignore]
  \node[Nmarks=if,iangle=-90,fangle=90](1)(0,0){}
  \node[Nmarks=f,iangle=-90,fangle=90](2)(20,0){}
  \drawloop[loopangle=180,ELdist=-1](1){$\begin{array}{c} a \mid 1 \\ b \mid 0 \end{array}$}
  \drawloop[loopangle=0,ELdist=-1](2){$\begin{array}{c} a \mid 0 \\ b \mid 1 \end{array}$}
  \drawedge(1,2){$\begin{array}{c} a \mid 1 \\ b \mid 1 \end{array}$}
\end{gpicture}
\begin{gpicture}[name=A12,ignore]
  \node[Nmarks=i](1)(0,0){}
  \node[Nmarks=f](2)(20,0){}
  \drawloop(1){$a\mid 1$}
  \drawloop(2){$a\mid 1$}
  \drawedge(1,2){$a\mid 1$}
\end{gpicture}
\begin{gpicture}[name=A34,ignore]
  \node[Nmarks=if](1)(0,0){}
  \node[Nmarks=if](2)(30,0){}
  \drawloop(1){$a\mid 2$}
  \drawloop[loopangle=-90](1){$b\mid 1$}
  \drawloop(2){$a\mid 1$}
  \drawloop[loopangle=-90](2){$b\mid 3$}
\end{gpicture}
\begin{gpicture}[name=A35,ignore]
  \node[Nmarks=if](1)(0,0){}
  \node[Nmarks=if](2)(30,0){}
  \drawloop(1){$a\mid 1$}
  \drawloop[loopangle=-90](1){$b\mid 0$}
  \drawloop(2){$a\mid 0$}
  \drawloop[loopangle=-90](2){$b\mid 1$}
\end{gpicture}
In this section, we give examples separating the classes of finitely,
polynomially and exponentially ambiguous aperiodic weighted automata for several
weight structures including the semiring of natural numbers $\N_{+,\times}$, the
max-plus semiring $\N_{\max,+}$ and the min-plus semiring $\N_{\min,+}$.

\begin{example}\label{ex:7.2}
  Let $\Sigma$ be any alphabet, $\Weights$ a set of weights, and
  $\A=(Q,\Sigma,\Delta,\wgt,I,F)$ any (possibly aperiodic) weighted automaton
  over $\Sigma$ and $\Weights$ which is not polynomially ambiguous.
  \begin{enumerate}
    \item Since the size of the multisets $\usem\A(w)$ is not polynomially
    bounded with respect to $|w|$, there can be no polynomially ambiguous
    weighted automaton $\B$ with $\usem\A=\usem\B$.
    
    \item Assume that $|\Delta|\leq|\Weights|$ and all transitions of $\A$ have
    different weights, and consider $\A$ as a weighted automaton over the semiring
    $(\mathcal{P}_\text{fin}(\Weights^*),\cup,\cdot,\emptyset,\{\varepsilon\})$,
    or, equivalently, as a non-deterministic transducer outputting the weights
    of the transitions.  Again, there can be no polynomially ambiguous weighted
    automaton $\B$ with $\sem\A=\sem\B$.
    
    \item For each $q\in Q$ and $a\in\Sigma$, the
    transitions $\delta=(q,a,p)\in\Delta$ ($p\in Q$) are enumerated as
    $\delta_1,\dots,\delta_m$ where $m$ is the degree of non-determinism for
    $q\in Q$ and $a\in \Sigma$.  Then put $\wgt(\delta_i)=i$, and let $\Weights$
    comprise all these numbers.  In comparison to 2., $|\Weights|$ might be
    considerably smaller than $|\Delta|$.  But, again, over the semiring
    $(\mathcal{P}_\text{fin}(\Weights^*),\cup,\cdot,\emptyset,\{\varepsilon\})$
    there is no polynomially ambiguous weighted automaton equivalent to $\A$.
    \qed
  \end{enumerate}
\end{example}

This shows that for suitable idempotent semirings and also for non-deterministic
transducers, there are aperiodic weighted automata for which there is no
equivalent polynomially ambiguous weighted automaton.   
Next we show that this is also the case for the semiring of natural numbers
$\N_{+,\times}$, the max-plus semiring $\N_{\max,+}$ and the min-plus semiring
$\N_{\min,+}$.

\begin{example}\label{ex:plustimes-exponentially-ambiguous}
  Let $\Sigma=\{a\}$ and consider the automaton $\A$ below over the 
  semiring $\N_{+,\times}$ of natural numbers.
  \begin{center}
    \gusepicture{A}
  \end{center}
  Note that the weighted automaton computes the sequence $(F_n)_{n\geq0}$
  of Fibonacci numbers $0,1,1,2,3,5,\cdots$.  More precisely, for any $n\in\N$,
  we have $\sem\A(a^n)=F_n$.

  \smallskip %
  Clearly, $\A$ is exponentially ambiguous and aperiodic with index 2.
  In \cite{Mazowiecki_2019}, it was shown that the Fibonacci numbers cannot be
  computed by copyless cost-register automata.
  Here, we prove that there is no aperiodic polynomially ambiguous weighted
  automaton $\B=(Q,\Sigma,\Delta,\wgt,I,F)$ with $\sem\A=\sem\B$.
  Suppose there was such a trimmed automaton $\B$.  
  
  First, consider any loop $q\xrightarrow{a^k}q$ with $k\geq 1$ of $\B$.
  Since $\B$ is aperiodic and SCC-unambiguous, hence unambiguous on the
  component containing $q$, as in
  Example~\ref{ex:maxplus-exponentially-ambiguous}, it follows that
  $(q,a,q)\in\Delta$.  Next, we
  claim $\alpha=\wgt(q,a,q)=1$. Indeed, suppose that
  $\alpha\geq 2$. Choose $m,\ell\geq 2$ minimal such that there is a path
  reading $a^m$ from $I$ to $q$ and a path for $a^\ell$ from $q$ to $F$.
  Considering, for $n\geq m+\ell$, the path $\rho_n\colon
  I\xrightarrow{a^m}q\xrightarrow{a^{n-m-\ell}}q\xrightarrow{a^\ell}F$, we
  obtain $\sem\B(a^n)\geq\wgt(\rho_n)\geq2^{n-m-\ell}$.  Since $F_n=o(2^{n})$,
  for $n$ large enough, we get $F_n<2^{-m-\ell}\cdot 2^{n}$, a contradition.

  So, in $\B$ all loops have weight 1. Hence there exists $K\in\N$ such that
  $\wgt(\rho)\leq K$ for all paths $\rho$ in $\B$. Consequently, if $\B$ is
  polynomially ambiguous of degree $d$, we have $\sem\B(a^n)\leq O(n^d)$
  for $n\in\N$. This yields a contradiction since 
  $F_n\sim\frac{1}{\sqrt{5}}\Big(\frac{1+\sqrt{5}}{2}\Big)^{n}$ grows exponentially.
  \qed
\end{example}

\begin{example}\label{ex:maxplus-exponentially-ambiguous}
  Let $\Sigma=\{a,b,c\}$ and consider the function $f_{\max}\colon\Sigma^{*}\to\N$
  defined as follows.  For a word $w=w_0 c w_1 c\dots c w_n$ with
  $w_0,\dots,w_n\in\{a,b\}^*$, we let
  $f_{\max}(w)=\sum_{i=0}^n\max\{|w_i|_a,|w_i|_b\}$.
  Over the max-plus semiring $\N_{\max,+}$, this function is realized by the 
  automaton $\A$ below.
  \begin{center}
    \gusepicture{A3}
  \end{center}
  Notice that $\A$ is aperiodic and not polynomially ambiguous.  We show that
  $f_{\max}$ cannot be realized over the max-plus semiring by a polynomially
  ambiguous and aperiodic weighted automaton.
  
  Notice that a similar automaton was considered in \cite{KlimannLMP04}, the
  only difference being that $c$-transitions have weight 1.  It was shown that
  the corresponding series cannot be realized over $\N_{\max,+}$ by a finitely
  ambiguous weighted automaton, be it aperiodic or not.  Here we want to
  separate exponentially ambiguous from polynomially ambiguous.  We prove this
  separation for aperiodic automata which makes some of the arguments in the 
  proof simpler (essentially we have self-loops instead of cycles). The 
  separation also holds if we drop aperiodicity.
  
  Towards a contradiction, assume that there was a polynomially ambiguous and
  aperiodic weighted automaton $\B=(Q,\Sigma,\Delta,\wgt,I,F)$ which
  realizes the function $f_{\max}$.  We assume $\B$ to be trimmed.  We start with
  some easy remarks.
  \begin{enumerate}
    \item  If there is a cycle $p\xrightarrow{u^{k}}p$ in $\B$ with 
    $u\in\Sigma^{+}$ and $k\geq1$ then $p\xrightarrow{u}p$. 
    
    Let $m\geq1$ be the aperiodicity index of $\B$.  For $\ell k\geq m$ we have
    $u^{\ell k},u^{\ell k+1}\in\Lang{\B_{p,p}}$.  Since $\B$ is polynomially
    ambiguous, these cycles lie in some SCC which is unambiguous. If the cycle 
    around $p$ reading $u^{\ell k}$ is not a prefix of the cycle reading $u^{\ell k+1}$ 
    then we have two different cycles reading $u^{\ell k(\ell k+1)}$, a contradiction. 
    Therefore, the cycle reading $u^{\ell k+1}$ is 
    $p\xrightarrow{u^{\ell k}}p\xrightarrow{u}p$.
  
    \item  Consider a looping transition $\delta=(p,v,p)$ in $\B$ with 
    $v\in\Sigma$. Then, $\wgt(\delta)\in\{0,1\}$.
    
    Since $\B$ is trimmed, there is an accepting run
    $p_1\xrightarrow{u}p\xrightarrow{w}p_2$ with $|uw|\leq 2|Q|$.  We deduce
    that for all $\ell\geq0$ there is an accepting run reading $uv^{\ell}w$ with
    weight at least $\wgt(\delta)\cdot\ell$. Since 
    $f_{\max}(uv^{\ell}w)\leq\ell+|uw|$, we deduce that $\wgt(\delta)\in\{0,1\}$.
  
    \item  If there is a path 
    $p\xrightarrow{a}p\xrightarrow{v}q\xrightarrow{b}q$ in $\B$ with 
    $v\in\{a,b\}^{*}$, then one of the two looping transitions has weight zero:
    $\wgt(p,a,p)=0$ or $\wgt(q,b,q)=0$. 
    
    Since $\B$ is trimmed, there are two runs $p_1\xrightarrow{u}p$ and
    $q\xrightarrow{w}p_2$ with $p_1\in I$ initial, $p_2\in F$ final and
    $|uw|\leq 2|Q|$.  We deduce that for all $\ell\geq0$ there is an accepting
    run reading $ua^{\ell}vb^{\ell}w$ with weight at least
    $\ell\cdot(\wgt(p,a,p)+\wgt(q,b,q))$.  Since
    $f_{\max}(ua^{\ell}vb^{\ell}w)\leq\ell+|uvw|$, we deduce that
    $\wgt(p,a,p)+\wgt(q,b,q)\leq1$.
  \end{enumerate}
  Let $n=|Q|$ be the number of states in $B$.  We show below that for each
  $k\geq1$, the word $w_k=a^{n}n^{n}(ca^{n}b^{n})^{k-1}$ admits at least $2^{k}$
  accepting runs in $\B$.  This implies that $\B$ is not polynomially ambiguous,
  a contradiction.
  
  Let $M=\max(\wgt(\Delta))$ be the maximal weight used in $\B$.  Notice that
  $M\geq1$.  Fix $k\geq1$ and let $N\geq 2knM$.  Define $u_0=a^{N}b^{n}$ and
  $u_1=a^{n}b^{N}$.  For each word $x=x_1\cdots x_k\in\{0,1\}^{k}$, define
  $w_x=u_{x_1}cu_{x_2}c\cdots cu_{x_k}$ and consider an accepting run $\rho_x$
  of $\B$ reading $w_x$ and realizing $f_{\max}(w_x)=kN$.  For each $1\leq j\leq
  k$, we focus on the subrun $\rho^{j}_x$ of $\rho_x$ reading $u_{x_j}$.  
  
  Assume that $x_j=0$.  Using the remarks above, we deduce that the prefix of
  $\rho^{j}_x$ reading $a^N$ is of the form
  \begin{equation}
    p_1\xrightarrow{a^{\ell_1}}p_1\xrightarrow{a}p_2
    \xrightarrow{a^{\ell_2}}p_2\xrightarrow{a} \cdots
    \xrightarrow{a}p_{m}\xrightarrow{a^{\ell_{m}}}p_{m}
    \label{eq:prefix-run}
  \end{equation}
  where $p_1,\ldots,p_{m}$ are pairwise distinct and
  $N=m-1+\ell_1+\cdots+\ell_{m}$.  Since looping $a$-transitions have weights in
  $\{0,1\}$, we deduce that $\wgt(\rho^{j}_x)\leq N+(2n-1)M$.  We claim that in
  $\rho^{j}_x$, some $a$-loop has weight 1.  If this is not the case, then
  $\wgt(\rho^{j}_x)\leq (2n-1)M$.  We deduce that $\wgt(\rho_x)\leq
  (k-1)(N+(2n-1)M)+(2n-1)M+(k-1)M=(k-1)N+(2nk-1)M$, but 
  $\wgt(\rho_x)=kN=f_{\max}(w_x)$, a contradiction with $N\geq 2knM$.  Let
  $(p_i,a,p_i)$ be some $a$-loop of weight 1 in $\rho^{j}_x$.  We replace the
  prefix of $\rho^{j}_x$ reading $a^N$ with
  $$
  p_1\xrightarrow{a^{i-1}}p_i\xrightarrow{a^{n-m+1}}p_i
  \xrightarrow{a^{m-i}}p_m
  $$
  to obtain a run $\hat{\rho}^{j}_x$ reading $a^{n}b^{n}$.
  The suffix of $\rho^{j}_x$ reading $b^{n}$ has a form similar to
  \eqref{eq:prefix-run}, having at least one $b$-loop since $n=|Q|$.  From the
  third remark above, all $b$-loops in $\rho^{j}_x$ have weight 0. We deduce 
  that $\hat{\rho}^{j}_x$ has one $a$-loop with weight 1 but all its $b$-loops 
  have weight 0.
  
  We proceed similarly when $x_j=1$ defining a run $\hat{\rho}^{j}_x$ reading 
  $a^{n}b^{n}$ where all $a$-loops have weight 0 and one $b$-loop has weight 1.
  Now, consider the run $\hat{\rho}_x$ obtained from $\rho_x$ by replacing 
  $\rho^{j}_x$ with $\hat{\rho}^{j}_x$ for each $1\leq j\leq k$. We see that 
  $\hat{\rho}_x$ is an accepting run for $w_k$. Also, if $x,y\in\{0,1\}^{k}$ 
  are different then $\hat{\rho}_x\neq\hat{\rho}_y$. Therefore, $\B$ has at 
  least $2^{k}$ accepting runs reading $w_k$, which concludes the proof.
  \qed
\end{example}

\begin{example}\label{ex:minplus-exponentially-ambiguous}
  Let $\Sigma=\{a,b,c\}$ and consider the function $f_{\min}\colon\Sigma^{*}\to\N$
  defined as follows.  For a word $w=w_0 c w_1 c\dots c w_n$ with
  $w_0,\dots,w_n\in\{a,b\}^*$, we let
  $f_{\min}(w)=\sum_{i=0}^n\min\{|w_i|_a,|w_i|_b\}$.
  Over the min-plus semiring $\N_{\min,+}$, this function is realized by the 
  automaton $\A$ depicted in Example~\ref{ex:maxplus-exponentially-ambiguous}
  which is aperiodic and not polynomially ambiguous.  
  It was shown in~\cite{mazowieckiR2018}, that in the min-plus semiring
  there is no polynomially ambiguous weighted automaton $\B$ with $\sem\A=\sem\B$.
  \qed
\end{example}

Next we wish to show that aperiodic polynomially ambiguous weighted automata
are strictly more expressive than aperiodic finitely ambiguous weighted
automata.

\begin{example}\label{ex:7.5}
  Let $\Sigma$ be any alphabet, $\Weights$ a set of weights and $\A$ an
  aperiodic polynomially ambiguous weighted automaton which is not finitely
  ambiguous.  We may argue as in Example~\ref{ex:7.2} to show that there is no
  finitely ambiguous weighted automaton $\B$ with $\usem\A=\usem\B$,
  respectively, under the assumptions of Example~\ref{ex:7.2}, with
  $\sem\A=\sem\B$ for the idempotent semiring
  $(\mathcal{P}_\text{fin}(\Weights^*),\cup,\cdot,\emptyset,\{\varepsilon\})$.
  \qed
\end{example}

We show that this is also the case for the semiring of natural numbers
$\N_{+,\times}$, the max-plus semiring $\N_{\max,+}$ and the min-plus semiring
$\N_{\min,+}$.

\begin{example}\label{ex:plustimes-polynomially-ambiguous}
  Consider the following automaton $\A$ over $\Sigma=\{a\}$ and the semiring
  $\N_{+,\times}$.
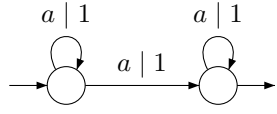
  \begin{center}
    \gusepicture{A12}
  \end{center}
  Clearly, $\sem\A(a^n)=n$ for each $n>0$, and $\A$ is aperiodic and
  polynomially (even linearly) ambiguous. But $\A$ is not equivalent to
  any finitely ambiguous weighted automaton.

  \smallskip %
  Towards a contradiction, suppose there was a trimmed finitely ambiguous
  weighted automaton $\B$ with $\sem\B=\sem{\A}$.

  \begin{remark*}
    Let $q\xrightarrow{a^m}q$ be a loop in $\B$ with weight $\alpha$, where
    $m\geq 1$. Then $\alpha=1$.
  \end{remark*}

  Indeed, choose a path in $\B$ from $I$ to $q$ with label $u$ and a path from
  $q$ to $F$ with label $v$. Then $\sem\B(u a^{m n}v)\geq\alpha^n$, for each $n\in\N$.
  On the other hand, $f(u a^{m n}v)=|u v|+m\cdot n$. Hence $\alpha\geq 2$ is
  impossible, showing $\alpha=1$.

  Consequently, in paths of $\B$ we may remove all loops without changing the
  weight.  Hence there is $C\in\N$ such that $\wgt(\rho)\leq C$ for each run
  $\rho$ of $\B$.  Since $\B$ is finitely ambiguous, it follows that
  $\{\sem\B(w)\mid w\in\Sigma^*\}$ is bounded.  This contradicts $\sem\B=\sem{\A}$.
  \qed
\end{example}

\begin{example}\label{ex:maxplus-polynomially-ambiguous}
  Consider the following automaton $\A$ over $\Sigma=\{a,b\}$ and $\N_{\max,+}$.
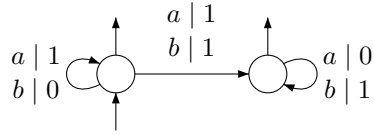
  \begin{center}
    \gusepicture{A5}
  \end{center}
  Note that $\A$ is almost identical to the automaton of
  Example~\ref{ex:minplus-polynomially-ambiguous}-2, used for $\N_{\min,+}$ in
  \cite{mazowieckiR2018}.
  Now for $f=\sem\A$ we have $f(w)=\max\{|u|_a+|v|_b\mid w=uv\}$ for
  each $w\in\Sigma^+$. Clearly, $\A$ is aperiodic and polynomially ambiguous.
  Now, we show that no aperiodic finitely ambiguous weighted
  automaton is equivalent to $\A$ over $\N_{\max,+}$.
  
  \smallskip %
  Suppose there was a trimmed weighted automaton $\B=(Q,\Sigma,\Delta,\wgt,I,F)$
  both aperiodic and finitely ambiguous, and with $\sem\B=f$.  We make
  the following observations on the structure of $\B$.
  
  \begin{myremark}\label{remark:1}
    If $\B$ contains a loop $q\xrightarrow{a^k}q$ for some $q\in Q$ and $k\geq 1$,
    then $t=(q,a,q)\in\Delta$, and the loop is a sequence of this transition $t$.
    
    This follows from the fact that $\B$ is aperiodic and unambiguous on the
    strong component containing $q$ (as in
    Example~\ref{ex:maxplus-exponentially-ambiguous}).
  \end{myremark}
  
  \begin{myremark}\label{remark:2}
    $\B$ cannot contain a path of the form
    $p\xrightarrow{a}p\xrightarrow{a^k}q\xrightarrow{a}q$ with $p\neq q$.
    
    Indeed, otherwise the word $a^{n+k}$ would have at
    least $n+1$ different paths from $p$ to $q$. Since $\B$ is trimmed, this
    contradicts the finite ambiguity of $\B$.
  \end{myremark}
  
  \begin{myremark}\label{remark:3}
    If $(q,a,q)\in T$ and $\alpha=\wgt(q,a,q)$, then $\alpha\in\{0,1\}$.
    
    Indeed,
    let $u$ be the label of a path from $I$ to $q$ and $v$ the label of a path
    from $q$ to $F$. Let $w_n=u a^n v$. Then $f(w_n)\leq |u v|+n$, and
    $\sem\B(w_n)\geq \alpha\cdot n$ for each $n\in\N$. This shows that $\alpha\leq 1$.
  \end{myremark}
  
  \begin{myremark}\label{remark:4}
    $\B$ cannot contain a path of the form
    $p\xrightarrow{b\mid 1}p\xrightarrow{v}q\xrightarrow{a\mid 1}q$ with 
    $v\in\Sigma^{*}$.
    
    Indeed, otherwise let $u$ be a label of
    a path from $I$ to $p$ and $w$ the label of a path from $q$ to $F$.
    Consider $w_n=u b^n v a^n w$ ($n\in\N$). Then $f(w_n)\leq |u v w|+n$ but
    $\sem\B(w_n)\geq 2n$, a contradiction for $n> |u v w|$.
  \end{myremark}
  
  \begin{lemma}\label{lem:uambmv}
    Let $m\geq |Q|$ and $u,v\in\Sigma^*$. Then $\B$ contains an accepting path
    for the word $u a^m b^m v$ of the form
    
    \begin{gpicture}[name=A8,ignore]
      \node[Nmarks=i](1)(0,0){$i$}
      \node(2)(20,0){$p$}
      \node(3)(40,0){$q$}
      \node[Nmarks=f](4)(60,0){$f$}
      \drawloop(2){$a\mid 1$}
      \drawloop(3){$b\mid 1$}
      \drawedge(1,2){$u a^{k_1}$}
      \drawedge(2,3){$a^{k_2} b^{k_3}$}
      \drawedge(3,4){$b^{k_4} v$}
    \end{gpicture}
    \centerline{\gusepicture{A8}}
    \noindent with $k_1,k_2,k_3,k_4<|Q|$.
  \end{lemma}
  \begin{proof}
    Let $n\geq m$ and $w_n=u a^n b^n v$. Then $f(w_n)\geq 2n$.
    Consider a path $\rho$ for $w_n$ in $\B$ with $\wgt(\rho)=f(w_n)$. The subpath
    of $\rho$ realizing $a^n$ must contain at least one $a$-loop, and by
    Remarks~\ref{remark:1} and~\ref{remark:2} it contains exactly one $a$-loop
    which is a power of a single transition.
    
    Hence $\rho$ has the form
    
    \begin{gpicture}[name=A9,ignore]
      \node[Nmarks=i](1)(0,0){$i$}
      \node(2)(20,0){$p$}
      \node(3)(40,0){$q$}
      \node[Nmarks=f](4)(60,0){$f$}
      \drawloop(2){$a\mid \alpha$}
      \drawloop(3){$b\mid \beta$}
      \drawedge(1,2){$u a^{k_1}$}
      \drawedge(2,3){$a^{k_2} b^{k_3}$}
      \drawedge(3,4){$b^{k_4} v$}
    \end{gpicture}
    \centerline{\gusepicture{A9}}
    
    \noindent with $k_1,k_2,k_3,k_4<|Q|$, and where the transition
    $(p,a,p)$ is taken $n-k_1-k_2$ times and the transition
    $(q,b,q)$ is taken $n-k_3-k_4$ times in $\rho$.
    
    By Remark~\ref{remark:3}, we have $\alpha,\beta\in\{0,1\}$.
    Let $\rho_1 \rho_2 \rho_3$ be the path obtained from $\rho$ by deleting the
    loops at $p$ and at $q$: $\rho_1=i\xrightarrow{ua^{k_1}}p$,
    $\rho_2=p\xrightarrow{a^{k_2}b^{k_3}}q$, and
    $\rho_3=q\xrightarrow{b^{k_4}v}f$.
    Let $c=\wgt(\rho_1 \rho_2 \rho_3)$. 
    Then $\wgt(\rho)\leq c+n\cdot\alpha+n\cdot\beta$.
    
    But $\wgt(\rho)=f(w_n)\geq 2n$.  Since $u,v\in\Sigma^*$ are fixed, there are
    only finitely many values $c=\wgt(\rho_1 \rho_2 \rho_3)\in\N$ which can
    arise in $\B$ as above with $i,p,q,f\in Q$ and $k_1,k_2,k_3,k_4<|Q|$,.  By
    choosing $n$ larger than their maximum, we obtain a path for $w_n=u a^n b^n
    v$ as above and now for this path it follows that $\alpha=\beta=1$.  By
    reducing the number of loops taken at $p$ and at $q$, we obtain an accepting
    path of the prescribed form for $w_m=u a^m b^m v$, proving the lemma.
  \end{proof}

  Now, let $m\geq|Q|$ and consider the word $w_K=(b^m a^m)^K$ ($K\in\N$).  For
  all $0<k<K$ we can write $w_K=u_k a^m b^m v_k$ with $u_k=(b^m a^m)^{k-1}b^m$
  and $v_k=a^m (b^m a^m)^{K-k-1}$.  We apply Lemma~\ref{lem:uambmv} to the word
  $u_k a^m b^m v_k$ and obtain a path $\rho_k$ of the form

  \begin{gpicture}[name=A10,ignore]
    \node[Nmarks=i](1)(0,0){}
    \node(2)(20,0){}
    \node(3)(40,0){}
    \node[Nmarks=f](4)(60,0){}
    \drawloop(2){$a\mid 1$}
    \drawloop(3){$b\mid 1$}
    \drawedge(1,2){$u_k a^{k_1}$}
    \drawedge(2,3){$a^{k_2} b^{k_3}$}
    \drawedge(3,4){$b^{k_4} v_k$}
  \end{gpicture}
  \centerline{\gusepicture{A10}}

  \noindent
  We claim that if $0<k<k'<K$, then $\rho_k\neq \rho_{k'}$.
  Indeed, if $\rho_k=\rho_{k'}$, we see that the path $\rho_{k'}$ must have the form
  
  \begin{gpicture}[name=A11,ignore]
    \node[Nmarks=i](1)(0,0){}
    \node(2)(20,0){}
    \node(3)(40,0){}
    \node(4)(60,0){}
    \node(5)(80,0){}
    \node[Nmarks=f](6)(100,0){}
    \drawloop(2){$a\mid 1$}
    \drawloop(3){$b\mid 1$}
    \drawloop(4){$a\mid 1$}
    \drawloop(5){$b\mid 1$}
    \drawedge(1,2){$u_k a^{k_1}$}
    \drawedge(2,3){$a^{k_2} b^{k_3}$}
    \drawedge[dash={1.5}0](3,4){}
    \drawedge(4,5){$a^{k_2'} b^{k_3'}$}
    \drawedge(5,6){$b^{k_4'} v_k'$}
  \end{gpicture}
  \centerline{\gusepicture{A11}}
  \noindent contradicting Remark~\ref{remark:4}. Therefore $\B$ contains at
  least $K-1$ accepting paths for $w_K$ ($K\in\N$). This contradicts $\B$
  being finitely ambiguous.
  
  We just note that by similar arguments and further analysing the weights of
  loops, it can be shown that $\A$ is not equivalent to any finitely ambiguous
  weighted automaton, even if it is not aperiodic.
  \qed
\end{example}

\begin{example}\label{ex:minplus-polynomially-ambiguous}
  Let $\Sigma=\{a,b\}$. 
  \begin{enumerate}
    \item Consider the following weighted automaton $\A$ over $\Sigma$ and
    $\N_{\min,+}$ from~\cite[p.558]{Kirsten_2008}:

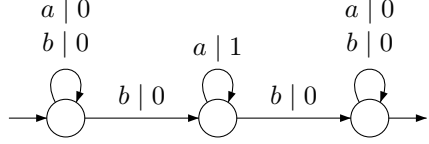
    \begin{gpicture}[name=A2,ignore]
      \node[Nmarks=i](1)(0,0){}
      \node(2)(20,0){}
      \node[Nmarks=f](3)(40,0){}
      \drawloop(1){$\begin{array}{c} a \mid 0 \\ b \mid 0 \end{array}$}
      \drawloop(2){$a \mid 1$}
      \drawloop(3){$\begin{array}{c} a \mid 0 \\ b \mid 0 \end{array}$}
      \drawedge(1,2){$b \mid 0$}
      \drawedge(2,3){$b \mid 0$}
    \end{gpicture}

    \centerline{\gusepicture{A2}}

    Here $\sem\A(w)$ is the least $\ell\geq 0$ such that $b a^\ell b$ is a
    factor of $w$.  If $w$ does not admit a factor of this form, than
    $\sem\A(w)=\infty$.  Clearly, $\A$ is SCC-unambiguous and aperiodic, but, as
    shown in~\cite[Proposition 3.2]{Kirsten_2008}, $\A$ is not equivalent to any
    finitely ambiguous weighted automaton.
    
    \item Consider the following weighted automaton $\A$ over $\Sigma$ and
    $\N_{\min,+}$ from~\cite{mazowieckiR2018}:

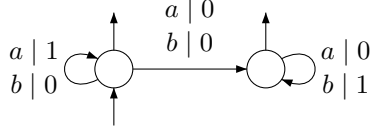
    \begin{gpicture}[name=A4,ignore]
      \node[Nmarks=if,iangle=-90,fangle=90](1)(0,0){}
      \node[Nmarks=f,iangle=-90,fangle=90](2)(20,0){}
      \drawloop[loopangle=180,ELdist=-1](1){$\begin{array}{c} a \mid 1 \\ b \mid 0 \end{array}$}
      \drawloop[loopangle=0,ELdist=-1](2){$\begin{array}{c} a \mid 0 \\ b \mid 1 \end{array}$}
      \drawedge(1,2){$\begin{array}{c} a \mid 0 \\ b \mid 0 \end{array}$}
    \end{gpicture}

    \centerline{\gusepicture{A4}}

    Then $\sem\A(w)=\min\{|u|_a+|v|_b \mid w=u v\}$. Clearly, $\A$ is aperiodic and
    polynomially ambiguous. As shown in~\cite[Example 15]{mazowieckiR2018},
    as a consequence of a pumping lemma, $\A$ is not equivalent to any
    finitely ambiguous weighted automaton.
    \qed
  \end{enumerate}
\end{example}

Finally, we wish to show that aperiodic finitely ambiguous weighted automata are
strictly more expressive than aperiodic unambiguous weighted automata.
Clearly, this can be derived for the idempotent semiring
$(\mathcal{P}_\text{fin}(\Weights^*),\cup,\cdot,\emptyset,\{\varepsilon\})$
as in Examples~\ref{ex:7.2} and \ref{ex:7.5}.  We show that this is also the
case for the semirings $\N_{+,\times}$, $\N_{\max,+}$ and $\N_{\min,+}$.

\begin{example}\label{ex:plustimes-finitely-ambiguous}
  Let  $\Sigma = \{a,b\}$  and consider the automaton $\A$ below
  over the semiring  $\N_{+,\times}$ of natural numbers.
  
  \centerline{\gusepicture{A34}}
  
  Clearly, $\A$  is aperiodic and 2-ambiguous, and
  $\sem\A(w) = 2^{|w|_a} + 3^{|w|_b}$  for each  $w \in \Sigma^*$.
  We show that no unambiguous weighted automaton is equivalent to $\A$.
  
  Suppose there was an unambiguous weighted automaton  $\B$  with $n$ states
  and with  $\sem\B = \sem\A$. Consider  $w = a^{n+2}b$.
  There is a unique successful path in $\B$  for $w$, having weight
  $\sem\B(w) = \sem\A(w) = 2^{n+2} + 3$. Then this path contains an
  $a$-loop $\rho$ of length $m \leq n$ and with  $\wgt(\rho) = C \in \N$.
  We have  $\sem\A(a^{n+m+2}b) = 2^{n+m+2} + 3$  and
  $\sem\B(a^{n+m+2}b) = C \cdot \sem\B(w) = C \cdot (2^{n+2} + 3)$.
  So  $2^{n+m+2} + 3 = C \cdot (2^{n+2} + 3)$. Then  $C < 2^m$.
  But  $(2^m - 1)\cdot(2^{n+2} + 3) < 2^{n+m+2} + 3$  as  $m \leq n$, a contradiction.  
\end{example}

\begin{example}\label{ex:minmaxplus-finitely-ambiguous}
  Let  $\Sigma = \{a,b\}$  and consider the automaton $\A$ below
  with weights in $\N$.
  
  \centerline{\gusepicture{A35}}
  
  Clearly, $\A$  is aperiodic and 2-ambiguous. Over the semiring $\N_{\max,+}$ we have
  $\sem{\A}_{\max}(w) = \max\{|w|_a, |w|_b\}$, and over the semiring $\N_{\min,+}$ we have
  $\sem{\A}_{\min}(w) = \min\{|w|_a, |w|_b\}$.
  As shown in \cite[p.255]{KlimannLMP04}, resp.\ \cite[Example 8]{mazowieckiR2018}, in both cases there is no
  unambiguous weighted automaton equivalent to  $\A$.  
\end{example}

\section{Conclusion}

We introduced a model of aperiodic weighted automata and showed that a suitable
concept of weighted first order logic and two natural sublogics have the same
expressive power as polynomially ambiguous, finitely ambiguous, resp.\
unambigous aperiodic weighted automata.  For the three semirings
$\N_{+,\times}$, $\N_{\max,+}$ and $\N_{\min,+}$ we showed that the hierarchies
of these automata classes and thereby of the corresponding logics are strict.

Our main theorem generalizes to the weighted setting a classical result of
automata theory.  A challenging open problem is to obtain similar results for
suitable weighted linear temporal logics.  Another interesting problem is to
characterize \wFO with unrestricted weighted products, possibly using
\emph{aperiodic} restrictions of the pebble weighted automata studied in
\cite{BGMZ-ICALP2010,KreutzerR2013,BolligGMZ14}.

Decidability problems for \wFO or equivalently for weighted aperiodic automata
are also open and very interesting.  For instance, given a \wMSO sentence, is
there an equivalent \wFO sentence?  Decidability may indeed depend on the
specific semiring.




\end{document}